\documentclass[a4paper,11pt]{article}\usepackage[]{graphicx}\usepackage[]{color}
\makeatletter
\def\maxwidth{ %
  \ifdim\Gin@nat@width>\linewidth
    \linewidth
  \else
    \Gin@nat@width
  \fi
}
\makeatother

\definecolor{fgcolor}{rgb}{0.345, 0.345, 0.345}

\usepackage{framed}
\makeatletter
 {\par\unskip\endMakeFramed%
 \at@end@of@kframe}
\makeatother

\definecolor{shadecolor}{rgb}{.97, .97, .97}
\definecolor{messagecolor}{rgb}{0, 0, 0}
\definecolor{warningcolor}{rgb}{1, 0, 1}
\definecolor{errorcolor}{rgb}{1, 0, 0}
\newenvironment{knitrout}{}{} 

\usepackage{alltt}
\usepackage[utf8x]{inputenc}
\usepackage{graphicx}
\usepackage{makeidx}
\usepackage{url} 
\usepackage{fancyvrb} 
\usepackage{lmodern}
\usepackage[T1]{fontenc}
\usepackage{amsmath,amssymb,amsfonts,amsthm}
\usepackage{url}
\usepackage{hyperref}
\usepackage{anysize,enumerate}
\usepackage{placeins}
\usepackage{tikz}
\usetikzlibrary{calc, fadings}
\usetikzlibrary{arrows}
\usetikzlibrary{patterns}
\marginsize{2.5cm}{2.5cm}{2cm}{2.5cm}
\usepackage{textcomp,array}
\usepackage[font=small, labelfont={sf,bf}, margin=1cm]{caption}
\usepackage{subcaption}
\usepackage {color,graphicx}
\usepackage{soul}
\usepackage{fancyhdr}
\renewcommand\appendix{\par
  \setcounter{section}{0}
  \setcounter{subsection}{0}
  \setcounter{figure}{0}
  \setcounter{table}{0}
  \renewcommand\thesection{Appendix \Alph{section}}
  \renewcommand\thefigure{\Alph{section}\arabic{figure}}
  \renewcommand\thetable{\Alph{section}\arabic{table}}
  
\usepackage{booktabs}
\renewcommand{\floatpagefraction}{.8}%
}

\title{{The doctrinal paradox:\\ ROC analysis in a probabilistic framework}
}
\author{Aureli Alabert \\ 
           Department of Mathematics \\
           Universitat Aut\`onoma de Barcelona \\  
           08193 Bellaterra, Catalonia \\  
           \url{Aureli.Alabert@uab.cat}  
           \and
           Mercè Farré \\ 
           Department of Mathematics \\
           Universitat Aut\`onoma de Barcelona \\  
           08193 Bellaterra, Catalonia \\  
           \url{farre@mat.uab.cat}
}

\newtheorem{teor}{Theorem}[section]
\newtheorem{propos}[teor]{Proposition}     
\newtheorem{defi}[teor]{Definition}
\newtheorem{remark}[teor]{Remark}

\newtheorem{lema}[teor]{Lemma}
\newtheorem{cor}[teor]{Corollary}

\newcommand\ct[4]
{
\fbox{
$\begin{matrix}
#1 & #2 \\
#3 & #4 
\end{matrix}$
}
}
\newcommand{\matri}[4]{$\begin{array}{cc}
  #1&#2\\#3&#4
  \end{array}$}
\IfFileExists{upquote.sty}{\usepackage{upquote}}{}
\begin{document}
\maketitle

\maketitle
\begin{abstract}
The \emph{doctrinal paradox} is analysed from a probabilistic point of view assuming a simple parametric model 
for the committee's behaviour. The well known \emph{issue-by-issue}   and \emph{case-by-case} 
 majority rules are compared in this model, by means of the concepts of \emph{false positive rate} 
 (FPR), \emph{false negative rate} (FNR)   
 and \emph{Receiver Operating Characteristics} (ROC) space.
We introduce also a new rule that we call \emph{path-by-path}, which is somehow halfway between the
other two.  
Under our model assumptions, the \emph{issue-by-issue} rule is shown to be the best of the three
according to an optimality criterion based in ROC maps, 
for all values of the model parameters (committee size and competence of its members),
when equal weight is given to FPR an FNR.
For unequal weights, the relative goodness of the rules depends on the values of the competence
and the weights, in a way which is precisely described. 
The results are illustrated with some numerical examples.


\end{abstract}

\noindent\textbf{Keywords}: Doctrinal paradox, judgement aggregation, ROC space. 

\section{Introduction: The doctrinal paradox}

 The \emph{Condorcet Jury Theorem} (attributed to Condorcet \cite{nla.cat-vn2403603}) states that 
 ``if $n$ jurists act independently, each with
 probability $\theta>\frac{1}{2}$  of making the correct decision, then the probability that the jury 
 (deciding by
majority rule) makes the correct decision increases monotonically to 1 as $n$ tends to infinity''. 
See for instance Boland \cite{10.2307/2348873}) and Karotkin and Paroush \cite{Karotkin2003}, 
and the references therein, for
precise statements, proofs, and extensions of this principle.

The \emph{doctrinal paradox} (a name introduced by Kornhauser \cite{Kornhauser1992169} in 1992) 
arises in some situations when a
committee or jury faces a compound question divided in two subquestions or premises, $P$
and $Q$.
 The point of interest is in deciding between the acceptance of
both premises $P\wedge Q$ ($P$ \emph{and} $Q$) and the
acceptance of the opposite $\neg(P\wedge Q)=\neg P \vee
\neg Q$ (\emph{not $P$ or not $Q$}). 
In view of the  Condorcet
Jury Theorem, some kind of majority rule seems appropriate for this two-premises problem.
However, in some cases, the same set of individual decisions leads to  different collective 
decisions depending 
on the manner the individual opinions are aggregated.

Classically, two standard decision procedures are considered in the literature:
the \emph{case-by-case} and the
\emph{issue-by-issue} procedures (\rm{CbyC} and \rm{IbyI}, respectively, for short), 
see \cite{Kornhauser1992169}, \cite{Kornhauser1993}. 
In CbyC, each committee member or judge
decides on both questions and votes  $P\wedge Q$ or
$\neg (P\wedge Q)$. Then, simple majority decides. It is also called 
the \emph{conclusion-based} 
or \emph{outcome-based} aggregation criterion. In fact, it is a reduction
to the one-premise Condorcet case.

In the \rm{IbyI} procedure, each
committee member decides  $P$ or $\neg P$ first, and then a joint
decision about this premise is taken by simple majority.
Similarly, each member chooses between $Q$ or $\neg Q$, and a joint decision is taken again 
by simple majority. If $P$ and $Q$ are separately decided by a (perhaps differently formed) 
majority, then $P\wedge Q$ is proclaimed. Otherwise, 
$\neg(P\wedge Q)$ is the conclusion.
 The IbyI procedure is also called the \emph{premise-based} aggregation criterion.

Both procedures look reasonable, but they may give rise to different results, hence the ``paradox'',
as we can see
in the examples of Tables \ref{tab1} and \ref{tab2}. The first one shows a 3-member committee, 
with a description of the individual member behaviour on the left, and a summary contingency 
table on the right. A
possible outcome of a 7-member decision is illustrated in Table \ref{tab2}. The contingency
table is enough to note the presence of the paradox. In both examples, the \rm{CbyC} and \rm{IbyI} procedures
lead to different decisions. 


\begin{table}[h] \centering
  {\extrarowheight 2pt
\begin{tabular}{c|c|c||c}
  Members & $P$ & $Q$ & $P\wedge Q$ \\
  \hline
  M$_1$ & $\checkmark$ &   $\times$ &   $\times$ \\
  M$_2$ &  $\times$  & $\checkmark$ &   $\times$ \\
  M$_3$ &$\checkmark$ & $\checkmark$ &  $\checkmark$ \\
  \hline
   & $\checkmark$  & $\checkmark$ &  $\times$ \\ 
\end{tabular}    
\qquad 
\begin{tabular}{r|cc|c} 
  \multicolumn{4}{c}{} \\  
  & $Q$ & $\neg Q$ & Totals \\
    \hline
    $P$ & 1 & 1 & 2 \\
    $\neg P$ & 1 & 0 & 1 \\
   \hline 
     Totals & 2 & 1 & $3$ 
  \end{tabular}
}
\caption{A doctrinal paradox with a 3-member
committee.
In the table on the left, symbols 
$\checkmark$ and $\times$ mean decision in favour or against the premise, respectively. 
 The contingency table on the right-hand
side summarizes the information.
The CbyC procedure
leads to decide $\neg(P\wedge Q)$; the IbyI rule leads to
$P\wedge Q$.} \label{tab1}
\end{table}

\begin{table}[h] \centering
    {\extrarowheight 2pt
  \begin{tabular}{r|cc|c}
     & $Q$ & $\neg Q$ & Totals \\
    \hline 
    $P$ & 3 & 2 & 5 \\
    $\neg P$ & 2 & 0 & 2 \\
   \hline 
     Totals & 5 & 2 & $7$ \\
  \end{tabular}
}
\caption{Contingency table showing the doctrinal paradox with a 7-member committee: 
  There are majorities for $P$ and $Q$ separately, but not a
conjoint majority for $P\wedge Q$.} \label{tab2}
\end{table}

The contingency table for the general case of $n$ committee members 
is illustrated in Table \ref{tab3}. We will assume throughout the paper 
that $n$ is an odd number: $n=2m+1$, $m\ge 1$. The doctrinal paradox appears when the
following conditions are simultaneously satisfied:
\begin{equation*} 
  x+y> m\ ,  \quad x+z > m\ , \quad x \leq m\ .
\end{equation*}

\begin{table}[h] \centering
    {\extrarowheight 2pt
  \begin{tabular}{r|cc|c}
     & $Q$ & $\neg Q$ & Totals \\
    \hline 
    $P$ & $x$ & $y$ & $x+y$ \\
    $\neg P$ & $z$ & $t$ & $z+t$ \\
   \hline 
     Totals & $x+z$ & $y+t$ & $n$ \\
  \end{tabular}
}
\caption{Summary of the distribution of $n$ votes on two premises as a contingency table.}\label{tab3}
\end{table}


We now introduce a new decision rule which is also reasonable and, as we will see later, lies in some sense
in between of the classical IbyI and CbyC rules. We call it 
\emph{path-by-path} 
(PbyP, for short), and it can be defined as follows: 
$P\wedge Q$ is proclaimed if the number of voters that individually decide $P\wedge Q$ is greater than those
who decide $\neg P$, and greater than those who decide $\neg Q$, separately. That is, in the notation of 
Table \ref{tab3}, 
if $x>z+t$ and $x>y+t$. 
Comparing with IbyI, only the number of supporters of  $P\wedge Q$ can be used to
 beat $\neg P$, without
 using the votes for $P\wedge\neg Q$, and similarly to beat $\neg Q$ without using the votes for $\neg P\wedge Q$;
it is therefore a stronger requirement to conclude $P\wedge Q$.
Comparing with CbyC, in order to conclude $P\wedge Q$, 
in PbyP the votes for $P\wedge Q$ do not need to beat the sum of 
all other options, but only those who deny $P$ and those who deny $Q$, 
separately, which is a weaker statement.

Our goal is to compare the performance of the three decision rules, 
for different committee sizes and different individual 
competence of its members. 
To this end, we define a theoretical framework consisting of a 
probabilistic model where the \emph{competence} 
of a judge is defined as the probability that he/she
takes the correct decision about each single premise. 
It is assumed that a ``true state of nature'' or ``absolute truth'' exists, which is one of the four possibilities 
that combine $P$, $Q$ and their negations. 

Our performance criterion is based in the concepts of true and false positive and negative rates
  and the 
 \emph{Receiver Operating Characteristics} (ROC) space. They have their origin in the field of electrical
 engineering and are commonly used in medicine, machine learning and other scientific disciplines 
 (see e.g. Fawcett \cite{Fawcett:2006:IRA:1159473.1159475} and 
 Hand and Till \cite{HandTill2001}). We believe that its application to the doctrinal paradox
 is completely new, and that with them the doctrinal paradox can be considered ``solved''.
 Of course, this claim has to be qualified: We mean that there is a fairly acceptable framework
 to decide which one of a set of rules has to be applied to arrive to the right conclusion.
 As will be apparent later, our analysis can be applied to any given set of rules, beyond those
 considered here.


\bigskip 

The problem of a committee deciding on the three logical clauses $P$, $Q$ and $R$, with 
the constrain $R\Leftrightarrow P\wedge Q$ is
only an instance of the broader situation in which a collective decision is to be built
from individual decisions in a community. The theory of \emph{judgement aggregation} aims at studying
and shedding light into these kind of problems. We refer the reader to the surveys  \cite{List2012} 
and \cite{List2009-LISJAA} for an overview of
the field and its recent developments. 

The doctrinal paradox is correspondingly a particular case of a general impossibility theorem
inside that theory (see, e.g. List and Pettit \cite{list_pettit_2002}): Under reasonable 
assumptions, there exist individual logically consistent decisions on $P$, $Q$ and $R$ that
lead to collective inconsistent decisions. One might think that the majority rule used
to construct the collective decision is to be held responsible of the impossibility;
however, except in trivial cases, any other aggregation rule (see \cite{List2012} for
examples) displays the same deadlock. 
Dietrich and List \cite{Dietrich2007} prove Arrow's impossibility result on preference 
aggregation as a corollary 
of this impossibility in judgement aggregation.
See also Camps, Mora and Saumell \cite{Camps2012} 
for a new approach to the problem of 
constrained judgement aggregation in a general setting.

The concept of decision rule that we introduce in the next section is somewhat narrower than
that of aggregation rule in judgement aggregation theory, but sufficient and adapted to
our purposes.
We do not go further  explaining judgement aggregation theory concepts since we focus 
specifically in the doctrinal paradox with a simple model of behaviour of the committee members.
For instance, we disregard strategic behaviour, considered in 
Dietrich and List \cite{dietrich_list_2007} and de Clippel and Eliaz \cite{deClippel201534}, 
or the epistemic or behavioural perspective, studied in
Bovens and Rabinowicz \cite{Bovens2004}, \cite{Bovens2006}), and Bonnefon \cite{Bonnefon2010}).

\bigskip

The paper is organised as follows: In Section  \ref{3rules}, 
we define and characterise with precision the
IbyI, CbyC and PbyP decision rules, and 
explain what we consider to be an \emph{admissible rule} in the application context we are dealing
with. 
We show that 
the three rules considered are admissible, and that there exist non-admissible (though not completely
irrational) decision rules.

The specific model assumptions are given in Section \ref{model}. 
Although the doctrinal paradox cannot be avoided, one can speak of the ``best rule'', 
once some theoretical model is defined and some reasonable performance criterion is chosen. 
Of course, different criteria gives rise to different ``best rules'', and this is again 
unavoidable.

The concept of true and false positives and negatives and that of ROC space 
are introduced in Section \ref{roc}.   
Translated to our setting, 
the false positive rate  $\text{\rm FPR}$ will be the probability of accepting $P\wedge Q$ when it is false, 
and the false negative rate 
$\text{\rm FNR}$ the probability of rejecting $P\wedge Q$ when it is true.

To make the exposition smooth, we postpone the proofs of all the technical statements of Sections \ref{3rules},
\ref{model} and \ref{roc}  
to an appendix.

Section \ref{res} contains the main results of the paper and their proofs: Rule  \rm{IbyI} 
is the best in the ROC setting under a 
symmetric criterion which gives the same weight to FPR and FNR; in case of unequal weights, 
any of the three rules can be the best, depending on the relation between the 
competence parameter and the specifics weights. 
  
Section \ref{sec:examples} contains some numerical computations and figures, showing that
all values of interest resulting from the probabilistic model can be explicitly obtained. 
More than that, the simple hypotheses on the model that we impose in Section \ref{model}
can be relaxed to a great extent and the explicit computations can still be carried
out without difficulty with adequate computing resources. This is explained in more detail in
the final discussion in Section \ref{disc}, together with other considerations and 
open problems.


\section{Decision rules}\label{3rules}

In this section, we give a detailed characterization of the IbyI, PbyP and  CbyC 
rules outlined in the introduction, and formalise the concept of admissible decision rule. 
We assume throughout the paper that the committee size is an odd number $n=2m+1$, with $m\ge 1$. The simple majority
for a single binary question is therefore achieved by any number of committee members greater than $m$.

\begin{defi}\label{def:3rules}
Assume that the opinions of the committee are summarised as in Table \ref{tab3}.
Then, we define the following decision rules:
\begin{enumerate}
\item[$R_1:$]  The \emph{IbyI (Issue-by-Issue)} rule, \\
\begin{equation}
\text{\rm \emph{Decide $P\wedge Q$ if and only if $x+y>z+t$ and $x+z>y+t$.}}
 \label{r1}
\end{equation}
\item[$R_2:$]  The \emph{PbyP (Path-by-Path)} rule,
\begin{equation}
 \text{\rm \emph{Decide $P\wedge Q$ if and only if $x>z+t$ and $x>y+t$.}}
 \label{r2}
\end{equation}
\item[$R_3:$]  The \emph{CbyC (Case-by-Case)} rule, 
\begin{equation}
 \text{\rm \emph{Decide $P\wedge Q$ if and only if $x>y+z+t$.}}
 \label{r3}
\end{equation}
\end{enumerate}
\end{defi}

In the sequel,
we shall use the following equivalent expressions, whose proof is straightforward and detailed in the Appendix.

\begin{propos}\label{pr1}  Using the notations in Table \ref{tab3}, we have:
\begin{align}
  &R_1\,\colon\ \text{\rm \emph{Decide $P\wedge Q$ if and only if $x>m-y\wedge z$.}}\label{r12}\\
  &R_2\,\colon\ \text{\rm \emph{Decide $P\wedge Q$ if and only if  $x>m-\big\lfloor \tfrac{y\wedge z}{2} \big\rfloor $.}}
  \label{r22}\\
  &R_3\,\colon\ \text{\rm \emph{Decide $P\wedge Q$ if and only if $x>m$.}}\label{r32}
\end{align}
where $\lfloor x\rfloor$ denotes the \emph{integer part} of $x$, i.e. the largest integer not greater than $x$, 
and $x\wedge y$ stands for the minimum of $x$ and $y$. \\
(The context will distinguish the uses of $\wedge$ as 
the minimum of two values or the logical operator `and'.)   
\end{propos}

From the characterisation of Proposition \ref{pr1}, it is clear
that the condition of rule $R_3$ to decide $P\wedge Q$ is more restrictive than that of $R_2$,
and the latter in turn is more restrictive that the condition of $R_1$. Furthermore, 
rules $R_2$ and $R_3$ are equivalent when $n=3\text{ or }5$, and they are different for $n\ge 7$. Rules
$R_1$ and $R_2$ are not equivalent for any $n\ge 3$. 
These facts will be stated as a proposition after a formal definition of decision rule:

\begin{defi}
  A \emph{decision rule} is a mapping 
from the set ${\mathbb T}$ of all contingency tables into $\{0,1\}$, where $1$ means deciding $P\wedge Q$, 
and $0$ 
means the opposite. 
\end{defi}

If the committee has $n$ members,
there are $N={{n+3} \choose {3}}$ ways to fill the contingency table, and $2^N$ possible decision rules. 
The number $N$ can be deduced by a combinatorial argument considering the number of ways
to express $n$ as the sum of four integers, including zero (the so-called \emph{weak compositions}
of a number). 

Since we assume that $P$ and $Q$ must have the same relevance in the final decision, it is natural to impose
that a decision rule must yield the
same result if we interchange the number of votes for 
$P\wedge\neg Q$ 
and $Q\wedge\neg P$%
, i.e. $y$ and $z$. 

  Furthermore, we would like to 
 consider only decision rules satisfying the following rationality property: If a table leads to
decision 1 and a committee member that has voted for $P\wedge\neg Q$ or $Q\wedge\neg P$ changes the vote
to $P\wedge Q$, the decision for the new table should also be 1; analogously, if the decision was 0 and 
the same vote changes to $\neg P \wedge \neg Q$, then the decision for the new table should also be 0.
This condition is easily implemented by considering only rules that preserve the 
partial order $\le$ on $\mathbb{T}$ generated by the four relations
\begin{alignat*}{3}
  \ct{x}{y}{z}{t}
  \ &\le \
  \ct{x+1}{y-1}{z}{t}
  \qquad\quad 
  &
  \ct{x}{y}{z}{t}
  \ & \le \
  \ct{x+1}{y}{z-1}{t}
  \\[5pt]
  \ct{x}{y}{z}{t}
  \ & \le \
  \ct{x}{y+1}{z}{t-1}
  \qquad\quad
  &
  \ct{x}{y}{z}{t}
  \ & \le \
  \ct{x}{y}{z+1}{t-1}
\end{alignat*}

These considerations lead us to define the following concept of admissible rule. More often than not, we will
use the vector notation $(x,y,z,t)$ instead of the tabular form, to save space.
\begin{defi}\label{def:AdmisRule}
  A decision rule $R\colon {\mathbb T}\longrightarrow \{0,1\}$ will be called an \emph{admissible rule} if:
\begin{enumerate}
  \item 
It does not distinguish between transposed tables:   
  \begin{equation*}
    R(x,y,z,t)=R(x,z,y,t) \ .
  \end{equation*}
  \item 
  It is order-preserving on the partially ordered set $(\mathbb{T},\le)$:
  \begin{equation*}
  (x,y,z,t)\le (x',y',z',t') \Rightarrow R(x,y,z,t)\le R(x',y',z',t')\ .
  \end{equation*}   
\end{enumerate}  
\end{defi}

The resulting partial order for $n=3$ is represented in Figure \ref{n3}, where we have already 
identified tables that merely interchange the values of $y$ and $z$.

\setlength{\abovecaptionskip}{10pt plus 0pt minus 0pt} 
\begin{figure}
\centering
\begin{tikzpicture}[->,>=stealth',shorten >=1pt,auto,node distance=1.25cm,
semithick]
\tikzstyle{every node}=[draw,rectangle, fill=black!10,rounded corners=3pt]
\scriptsize
\renewcommand\arraystretch{0.8}
\setlength\arraycolsep{0.3\tabcolsep}
\node (A3) {\matri 0 0 0 3};
\node (B3) [below of=A3] {\matri 0 1 0 2};
\path (A3) edge (B3);
\node (C3) [below of=B3] {\matri 0 1 1 1};
\node (C2) [left of=C3] {\matri 0 2 0 1};
\node (C4) [right of=C3] {\matri 1 0 0 2};
\path (B3) edge (C2) edge (C3) edge (C4);
\node (D3) [below of=C3] {\matri 1 1 0 1};
\node (D2) [left of=D3] {\matri 0 2 1 0};
\node (D1) [left of=D2] {\matri 0 3 0 0};
\path (C2) edge (D1) edge (D2) edge (D3);
\path (C3) edge (D2) edge (D3);
\path (C4) edge (D3);
\node (E2) [below of=D2] {\matri 1 2 0 0};
\node (E3) [below of=D3] {\matri 1 1 1 0};
\node (E4) [right of=E3] {\matri 2 0 0 1};
\path (D1) edge (E2);
\path (D2) edge (E2) edge (E3);
\path (D3) edge (E2) edge (E3) edge (E4);
\node (F3) [below of=E3] {\matri 2 1 0 0};
\path (E2) edge (F3);
\path (E3) edge (F3);
\path (E4) edge (F3);
\node (G3) [below of=F3] {\matri 3 0 0 0}; 
\path (F3) edge (G3);
\end{tikzpicture} 
\caption{The partially ordered set $(T,\le)$ for $n=3$. We have identified tables which are 
  transposed of each other} 
\label{n3}
\end{figure}
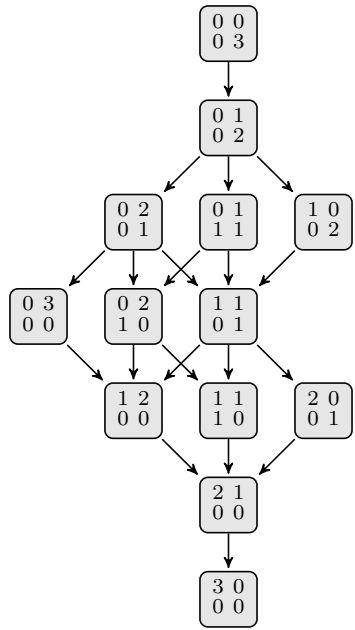

We will write $R\le R'$ whenever $R(T)\le R'(T)$ for all tables $T\in\mathbb{T}$,
and $R<R'$ whenever $R\le R'$ and $R\neq R'$.
Now, the relations among the rules $R_1$, $R_2$, $R_3$
can be stated as follows:

\begin{propos} \label{pr2}
 Rules $R_1$, $R_2$, $R_3$ of Definition \ref{def:3rules} are
 admissible and, as functions $ \mathbb{T} \rightarrow\{0,1\}$, they satisfy $R_3\le R_2\le R_1$. \\
 Moreover, we have 
 $R_3=R_2<R_1$ for $n=3,5$, and $R_3<R_2<R_1$ for $n\geq 7$.
 
\end{propos}

As an example of a non-admissible rule, consider requiring the simple majority of $P\wedge Q$ against
each one of the other options: 
\begin{equation*}
R_0(x,y,z,t)=1 \quad\Longleftrightarrow\quad x>y \text{ and } x>z \text{ and } x>t\ , 
\end{equation*}
Indeed, with $n\ge 5$, one has $(2,1,1,1)<(2,2,1,0)$, but applying $R_0$ to both sides 
reverses the inequality. This contradicts the second condition of Definition \ref{def:AdmisRule}.

The relation of $R_0$ with the other three rules can be summarised as follows:
In general, $R_3\le R_2\le R_0$. For $n=3$, the relations $R_3=R_2=R_0<R_1$ hold true, with $R_0$ and $R_1$ 
differing on $(1,1,1,0)$. 
For $n=5$, we have $R_3=R_2< R_0< R_1$ with $R_2$ and $R_0$ differing on $(2,1,1,1)$, and $R_0$ and $R_1$ 
differing on $(1,2,2,0)$ and on $(2,2,1,0)$.
Starting with $n\ge 7$, $R_0$ and $R_1$ are no more comparable (neither $R_0\le R_1$ nor $R_1\le R_0$), 
and all four rules are different.
All these relations are easily checked.

\begin{remark}
We could have started by defining an equivalence relation $\sim$ on ${\mathbb T}$,  
identifying tables 
\begin{equation*}
\ct{x}{y}{z}{t}
\quad
\text{and}
\quad
\ct{x}{z}{y}{t}
\end{equation*}
 The partial order would be shorter to define, and an admissible rule considered on the 
 quotient set ${\mathbb T}/\sim$ would not need condition 1 of the definition. Figure 1
 implicitly uses this equivalence.	
 Although elegant, this setup would introduce more complications in the discussions of 
 the next sections.  
\end{remark}  

\section{Probabilistic model of committee voting}\label{model}
Our probabilistic model is the simplest possible, 
and it is described by the four conditions below. 
Under these conditions, we can develop the theory of ROC optimality without unnecessary 
complications, and produce comprehensible 
examples. As we will discuss in Section \ref{disc}, these conditions can be very much relaxed and the
computations of the ROC analysis can be carried out automatically without problems.

A framework similar to ours can be found in List \cite{List2005}, where the main goal is to
compute the probability of appearance of the paradox, and to investigate the behaviour of this probability
when the committee size grows to infinity, in the spirit of the classical Condorcet theorem.


We assume that a true ``state of nature'' exists, in which one of the four exclusive 
events $P\wedge Q$, $P\wedge \neg Q$, $\neg P\wedge Q$ and  $\neg P\wedge \neg Q$ is in force.

 We assume the following conditions: 
\begin{enumerate}
  \item[(C1)] 
  \emph{Odd committee size}: The number of voters is an odd number, $n=2m+1$, with $m\geq 1$.
   \item[(C2)] \emph{Equal competence}:
   The probability $\theta$ of choosing the correct alternative when deciding 
   between $P$ and $\neg P$
   is the same for all voters and satisfies $\frac12< \theta<1$.
   The same competence $\theta$ is assumed when deciding between $Q$ and $\neg Q$.
   \item[(C3)] \emph{Mutual independence among voters}: The decision of each voter does
   not depend on the decisions of the other voters.
   \item[(C4)] \emph{Independence between $P$ and $Q$}: For each voter, the decision on one premise
   does not influence the decision on the other. 
\end{enumerate}

  Formally, conditions (C2)--(C4) can be rephrased by saying that for each voter $k$ in the committee 
and each clause $c\in\{P,Q\}$, there is a random variable 
that takes the value 1 if the voter believes the clause is true, and zero otherwise, and 
all these variables are stochastically independent and identically distributed. 
Their specific distribution depends on the true state
of nature.

Under these hypotheses, we can obtain the probability of all possible distribution 
of votes in a table  for each given state of nature.
Assume that $(X,Y,Z,T)$ are the random variables representing the counts in Table \ref{tab3} 
in the probabilistic framework just defined. As it is customary in probability and statistics, 
we distinguish between random variables represented by capital letters $X, Y$, etc, 
and their observed values, represented by small letters $x,y$, etc. 

\begin{propos}\label{pr3} Under conditions (C1)--(C4), the joint distribution of
$(X,Y,Z,T)$ is multinomial $M(n,p_x,p_y,p_z,p_t)$, with parameters 
depending on the true state of nature, according to the following table:

\centering
  \setlength\extrarowheight{2pt}
  \begin{tabular}{r|c|c|c|c}
             & $p_x$ & $p_y$ & $p_z$ & $p_t$ \\
  \hline  
  $P\wedge Q$ & $\theta^2$ & $\theta(1-\theta)$&  $\theta(1-\theta)$ & $(1-\theta)^2$ \\
  \hline 
  $P\wedge \neg Q$ & $\theta(1-\theta)$ & $\theta^2$ & $(1-\theta)^2$ & $\theta(1-\theta)$ \\
  \hline
  $\neg P\wedge Q$ & $\theta(1-\theta)$ & $(1-\theta)^2$ & $\theta^2$ & $\theta(1-\theta)$ \\
  \hline 
  $\neg P\wedge \neg Q$ & $(1-\theta)^2$ & $\theta(1-\theta)$ & $\theta(1-\theta)$ & $\theta^2$ \\
  \hline 
\end{tabular}
\end{propos}


The proof follows immediately from the definition of the multinomial distribution 
and the model assumptions, see the Appendix. Recall that
the multinomial probability function is given by
\begin{equation} 
\mathbb{P}\{X=k,Y=j,Z=i,T=\ell\}={n\choose i,j,k,\ell} p_x^k\cdot p_y^j\cdot p_z^i\cdot p_t^{\ell},  
\label{mult}
\end{equation}    
where $k,j,i,\ell$ are non-negative integers such that $n=i+j+k+\ell$,
and ${n\choose i,j,k,\ell}$ means the quotient of factorials 
$\frac{n!}{k!\cdot j!\cdot i!\cdot\ell!}$. 
 
From the law of $(X,Y,Z,T)$, it is easy to compute the law of any given decision rule 
$R\colon {\mathbb T}\rightarrow \{0,1\}$: It will be a Bernoulli law whose parameter is simply
the sum of the numbers (\ref{mult}) for all tables that $R$ maps to 1.

Notice that formula (\ref{mult}) is correct irrespective of the existence of a background absolute truth
or of the competence concept. It only needs independence between voters, and the existence of
a vector of probabilities $(p_x,p_y,p_z,p_t)$ adding up to 1, the same for all voters, representing
the probability of opting for each of the four options. 
  List \cite{List2005} studies the probability of appearance of the doctrinal paradox in this
  more general situation and shows that slightly different values of the vector of probabilities may lead 
  to very different 
  values of the probability of appearance of the paradox when $n\to\infty$. 
  Applied to our case, his results imply that, if $P\wedge Q$ is true, the probability of appearance
  of the paradox (disagreement between IbyI and CbyC rules) tends to 0 when the competence $\theta$ is greater than $\sqrt{0.5}$, and tends
  to 1 when it is lower; if $P\wedge Q$ is not true, then it always tends to 0.
  Interestingly, he also computes
  the expectation of appearance of the paradox when the vector of probabilities is assumed
  to follow a non-informative uniform prior on the simplex.
  
\section{True and false rates and ROC analysis}\label{roc}

\emph{Receiver operating characteristics}  (ROC) plots were introduced 
to visualize and compare binary classifiers in signal detection (see, e.g. Egan \cite{Egan1975})  
and its use extends to medical tests, machine learning 
and other disciplines where binary decisions have to be taken under uncertainty  
(see Fawcett \cite{Fawcett:2006:IRA:1159473.1159475} for an introductory presentation of ROC plots). 
The term \emph{classifier} is also used as a synonym of \emph{decision rule}.  


In signal detection theory, propositions are related to the emission/reception of a binary digit.
Denote by $\mathbf{\hat{0}}$ and $\mathbf{\hat{1}}$ the bit received and by $\mathbf{0}$
and $\mathbf{1}$ the bit actually sent.
The \emph{true positive rate} (TPR) is 
defined as the probability of receiving $\mathbf{\hat{1}}$ when $\mathbf{1}$ is the true bit emitted,
and the \emph{true negative rate} (TNR) as the probability of receiving $\mathbf{\hat{0}}$ when $\mathbf{0}$ 
is the bit sent. 
Analogously, the 
\emph{false positive rate} (FPR) and the \emph{false negative rate} (FNR) are, respectively, the probabilities of 
receiving 
$\mathbf{\hat{1}}$ when $\mathbf{0}$ is the true digit, and of receiving 
$\mathbf{\hat{0}}$ when $\mathbf{1}$ is the true digit. 
From these definitions, it is clear that  a decision rule such that $\text{\rm TPR}\approx 1$ and
and $\text{\rm FPR} \approx 0$ has a ``good performance''.

 \begin{remark}\label{rk3}
In classical statistics, decision rules appear in the context of \emph{hypothesis testing}, 
where the TPR is the \emph{power} of the test, the greater the better, 
under the restriction that the FPR (called the \emph{type-I error}) does not exceed a fixed small value 
(the \emph{significance level}).
The \emph{type-II error} corresponds to the FNR. The two types of errors are thus treated in 
a non-symmetric way.
In medicine, the $\text{\rm TPR}$ and the $\text{\rm TNR}$ are respectively called \emph{sensitivity} and 
\emph{specificity}. 
\end{remark}

In the \emph{ROC graph}, several classifiers can be compared
on the basis of the pair $(\text{\rm FPR},\text{\rm TPR}$) represented in the unit square $[0,1]\times[0,1]$, the
so-called \emph{ROC space} (see Figure \ref{roc1}). Usually, the rates are estimated
from sample data. ``Good'' decision rules are expected to correspond
to points close to the upper left corner
$(0,1)$ of the unit square. Different measures of the proximity to
that corner can be considered. The most widely used is the area of
the shaded triangle (AOT) in Figure \ref{roc1}, defined by the points  $(0,0)$,  $(1,1)$ and  
$( \mbox{\rm FPR}, \mbox{\rm TPR})$. 
The closer the area to 0.5, the better the classifier is considered. 
Points on the diagonal of the square correspond to completely random
classifiers, for which the probability of true and false positives are equal. Points below the
diagonal line represent classifiers that perform worse than
random. 
Following   \cite{Fawcett:2006:IRA:1159473.1159475}, the classifiers plotted near the $(0,0)$ corner can be said to be  
  ``conservative'',  because they make few positive classifications (true or false). For the same reason, 
  classifiers plotted near to the $(1,1)$ corner are sometimes called ``liberal''  because they tend
to have a higher number of false positives.

It is immediate from Figure \ref{roc1} that the area of the triangle    
can be expressed in terms of the rates as follows:
\begin{align}
 \mbox{\rm AOT} &= \tfrac{1}{2}(\mbox{\rm TPR}-\mbox{\rm FPR}) \label{aot2}\\
                &= \tfrac{1}{2}-\big(\tfrac{\text{FPR}}{2} + \tfrac{\text{FNR}}{2} \big)\ .\label{aot}
\end{align}



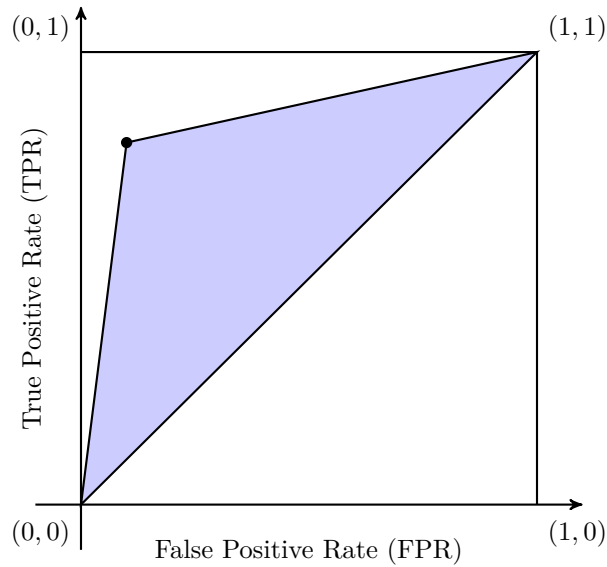
\begin{figure}[h] 
  \centering 
  \begin{tikzpicture}[
  thick,
  >=stealth',
  scale=6,
  every node/.append style={font=\small}
  ]
  \colorlet{col1}{blue!20}
  \coordinate (O) at (0,0);
  \fill [fill=col1] (0,0) -- (1,1) -- (0.1,0.8) -- cycle;
  \draw[->] (-0.1,0) -- (1.1,0) ;
  \draw[->] (0,-0.1) -- (0,1.1);
  \draw[-] (1,0) -- (1,1) -- (0,1);
  \draw [-, thick] ( 0,0) -- ( 1, 1) -- (0.1,0.8) -- cycle; 
  \draw [fill] (0.1,0.8) circle (0.01);
  \node [below] at (0.5, -0.05) {False Positive Rate (FPR)};
  \node [above, rotate=90] at (-0.05, 0.5) {True Positive Rate (TPR)};
  \node [below left] at (-0, -0.01) {$(0,0)$};
  \node [above right] at (1.00,1) {$(1,1)$};
  \node [below right] at (1.00,-0.01) {$(1,0)$};
  \node [above left] at (-0.00,1) {$(0,1)$};
  \end{tikzpicture}   
  \caption{The black dot represents a classifier in the ROC
    space, represented by its coordinates $(\text{\rm FPR},\text{\rm TPR})$, and the shaded area is 
    $ \text{\rm AOT}$.}\label{roc1}
\end{figure}

In the definition of AOT, the roles of the rates FPR and FNR are symmetric. 
In some situations, it may be desirable to assign different weights to these 
errors. This leads to the concept of 
\emph{weighted area} of the triangle, WAOT. 
Indeed, fixing a weight value $w\in(0,1)$, one can define, by analogy with formula (\ref{aot}),
\begin{equation}
\text{WAOT}_w:=  
\tfrac{1}{2}- \big( w\cdot\text{FPR} + (1-w)\cdot\text{FNR}\big) \ .  \label{WAOT}
\end{equation}

For any $w\in(0,1)$, $\text{WAOT}_w$ takes values in $[-\frac{1}2,\frac12]$, negative ``areas'' corresponding
to points below the diagonal. 
If $w>\frac12$, the weighted area $\text{WAOT}_w$  penalizes false positives more than  false negatives;  
and if $w<\frac12$, it is the other way round. 
The points of the ROC space yielding the same value of WAOT are straight lines, 
with slope equal to $w/(1-w)$, see Figure 
\ref{fig:waot}.

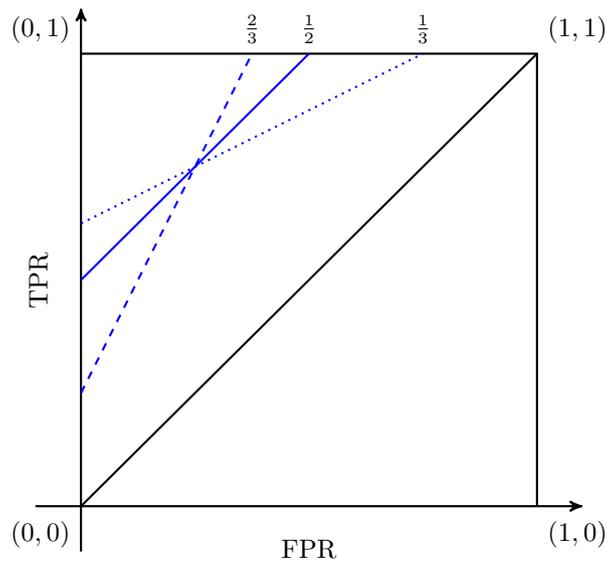
\begin{figure}[h] 
  \centering 
  \begin{tikzpicture}[
  thick,
  >=stealth',
  scale=6,
  every node/.append style={font=\small}
  ]
  \coordinate (O) at (0,0);
  \draw [help lines, pink] (0,0) grid (1, 1);
  \draw[->] (-0.1,0) -- (1.1,0) ;
  \draw[->] (0,-0.1) -- (0,1.1);
  \draw[-] (1,0) -- (1,1) -- (0,1);
  \draw [-, thick] ( 0,0) -- ( 1, 1); 
  \node [below] at (0.5, -0.05) {FPR};
  \node [above, rotate=90] at (-0.05, 0.5) {TPR};
  \node [below left] at (-0, -0.01) {$(0,0)$};
  \node [above right] at (1.00,1) {$(1,1)$};
  \node [below right] at (1.00,-0.01) {$(1,0)$};
  \node [above left] at (-0.00,1) {$(0,1)$};
  \draw [-, thick, blue] ( 0,0.5) -- ( 0.5, 1);
    \node [above, font=\footnotesize] at (3/8,1) {$\frac23$};  
  \draw [-, thick, blue, dashed] ( 0,0.25) -- ( 3/8, 1);  
    \node [above, font=\footnotesize] at (0.5,1) {$\frac12$};  
  \draw [-, thick, blue, dotted] ( 0,5/8) -- ( 0.75, 1); 
    \node [above, font=\footnotesize] at (0.75,1) {$\frac13$};

  \end{tikzpicture}   
  \caption{Level lines for the value $\text{WAOT}_w=0.25$ with weights $w=\frac23,\frac12,\frac13$.
  \label{fig:waot}}
\end{figure}

Unequal weights are useful in some practical situations:  
For instance, in court of justice cases, it is common that false positives 
(declaring guilty an innocent 
defendant) are considered worst than false negatives;  
in medical tests, the two errors often play an obvious asymmetric role too.

In some applications, the rates $\text{\rm TPR}$, $\text{\rm FPR}$ of a given classifier can be estimated on the basis 
of a ``test sample'' in which the actual states of the nature  are known ($\mathbf{0}$ or $\mathbf{1}$ 
in each observation) and the outputs of the classifier  ($\mathbf{\hat 0}$ or $\mathbf{\hat1}$) 
are  compared against the actual states. The results are often summarized in a contingency table known as 
\emph{confusion matrix}. 
In social applications, as is the case of court cases, the actual states are supposed to be unknown and there might not be
test samples 
available. However,
the rates $\mbox{\rm FPR}$ and $\mbox{\rm TPR}$ can be defined and computed exactly under 
our model assumptions as we will see in rest of this section. First, 
we translate the  $\mbox{\rm ROC}$ analysis vocabulary to our probabilistic framework. 

\begin{defi}\label{defi2}  
  Assume the model conditions  (C1)--(C4) of Section \ref{model}. 
The    true positive rate  associated to a decision rule $R$ is defined as the probability to  
``decide $P\wedge Q$'' under the state of the nature  ``$ P\wedge Q$ true'', and it depends on $n$, 
$\theta$ and the decision rule:
\begin{equation}
\text{\rm TPR}(n,\theta,R):= \mathbb{P}_{P\wedge Q} \{   R(X,Y,Z,T)=1\}, \label{tpr} 
\end{equation}
where $\mathbb{P}_{P\wedge Q}$  denotes the multinomial law of Proposition \ref{pr3}, 
with parameters corresponding to the state of nature ``$ P\wedge Q$ true''.
\end{defi}

\begin{propos}\label{pr5} 
Under the assumptions (C1)--(C4), the true positive rates defined in (\ref{tpr})  for rules $R_1$, $R_2$, $R_3$  are:
\begin{align}
& \text{\rm TPR}(n,\theta,R_1) = \sum_{i=0}^n  \sum_{j=0}^{n-i} \sum_{k=m- i\wedge j+1}^{n-i-j}
{n\choose i,j,k,\ell}\theta^{i+j+2k} (1-\theta)^{2n-i-j-2k}  \ , \label{R1tpr} \\
&   \text{\rm TPR}(n,\theta,R_2) =  \sum_{i=0}^n  \sum_{j=0}^{n-i} \sum_{k=m-\lfloor\frac{i\wedge j }{2}\rfloor+1}^{n-i-j}
{n\choose i,j,k,\ell}\theta^{i+j+2k} (1-\theta)^{2n-i-j-2k} \ , \label{R2tpr}  \\
&   \text{\rm TPR}(n,\theta,R_3) =  
\sum_{i=0}^n  \sum_{j=0}^{n-i} \sum_{k=m+1}^{n-i-j}
{n\choose i,j,k,\ell}\theta^{i+j+2k} (1-\theta)^{2n-i-j-2k} \ ,\label{R3tpr} 
\end{align} 
where $\ell=n-i-j-k$.
\qed
\end{propos}

See the Appendix for the proof.
Notice that the inequalities
$R_3\le R_2\le R_1$ induce the corresponding inequalities among the true positive rates: 
\begin{equation*}
  \text{\rm TPR}(n,\theta,R_3)\le  \text{\rm TPR}(n,\theta,R_2)\le \text{\rm TPR}(n,\theta,R_1)\ .
\end{equation*}

False positives can arise under the three different states of nature contained in the negation $\neg (P\wedge Q)$. 
To define the false positive rate FPR, we adopt the conservative approach, taking the maximum of the
probabilities of accepting $P\wedge Q$ under each of the states. 
We see in the next proposition that the largest probability always corresponds to the 
case when one of the clauses $P$ or $Q$ is true and the other one is false.  

\begin{propos}\label{pr6}
  Let $R$ be any one of the rules $R_1$, $R_2$ or $R_3$.
  Under assumptions (C1)--(C4) of Section \ref{model}, we have  
  \begin{equation}
    \mathbb{P}_{ P\wedge \neg Q} \{ R(X,Y,Z,T)=1 \}=  
    \mathbb{P}_{\neg P\wedge Q} \{  R(X,Y,Z,T)=1 \} >
    \mathbb{P}_{\neg P\wedge \neg Q} \{  R(X,Y,Z,T)=1 \}\ .  \label{eqpr6-a} 
  \end{equation}
\qed  
\end{propos}  

The equality in (\ref{eqpr6-a}) follows from condition (C2), which assumes equal 
competence on judging  $P$ and $Q$. The inequality on the right is intuitive
noticing that the state $\neg P\wedge \neg Q$ is 
``the less likely one'' to decide $P\wedge Q$. 
A formal proof is given in the Appendix. 

\begin{defi}\label{def:FPR}  
  Let $R$ be any one of the rules $R_1$, $R_2$ or $R_3$. 
  We define the false positive rate  as:
  \begin{align}
  \text{\rm FPR}(n,\theta,R):= &\mathbb{P}_{P\wedge \neg Q} \{ R(X,Y,Z,T)=1 \}\ .\label{fpr}
  \end{align} 
\end{defi}
\begin{defi}\label{def:AOT-WAOT}
    Let $R$ be any one of the rules $R_1$, $R_2$ or $R_3$. We define the area of the triangle as:
  \begin{equation}
    \text{\rm AOT}(n,\theta,R):=  \tfrac{1}{2} \big(\text{\rm TPR}(n,\theta,R)-\text{\rm FPR}(n,\theta,R) \big)\ .  
  \label{aotc}
  \end{equation}

  Fix $w\in(0,1)$. We define the weighted area of the triangle as:
  \begin{equation*}
  \text{\rm WAOT}_w(n,\theta,R):=  
  \tfrac{1}{2}- \big( w\,\text{\rm  FPR}(n,\theta,R) +  (1-w)\text{\rm FNR}(n,\theta,R)\big) \ .  \label{WAOT2}
  \end{equation*}
\end{defi}

Using the formulae of Proposition \ref{pr3}, one can write the analogue of Proposition \ref{pr5} for FPR.
\begin{propos}\label{pr7} Under the assumptions  (C1)--(C4), the false positive rate defined in 
  (\ref{fpr})  for each fixed rule  is:
  \begin{align}
    & \text{\rm FPR}(n,\theta,R_1) = \sum_{i=0}^n  \sum_{j=0}^{n-i} \sum_{k=m- i\wedge j+1}^{n-i-j}
    {n\choose i,j,k,\ell} \theta^{n-i+j} (1-\theta)^{n+i-j} \ , \label{R1fpr} \\
    &   \text{\rm FPR}(n,\theta,R_2) =  \sum_{i=0}^n  \sum_{j=0}^{n-i} 
    \sum_{k=m-\lfloor\frac{i\wedge j }{2}\rfloor+1}^{n-i-j}
    {n\choose i,j,k,\ell}\theta^{n-i+j} (1-\theta)^{n+i-j} \ , \label{R2fpr}  \\
    &   \text{\rm FPR}(n,\theta,R_3) =  \sum_{i=0}^n  \sum_{j=0}^{n-i} \sum_{k=m+1}^{n-i-j}
    {n\choose i,j,k,\ell}\theta^{n-i+j} (1-\theta)^{n+i-j} \ ,
    \label{R3fpr} 
  \end{align} 
where $\ell=n-i-j-k$. \qed  
\end{propos}
  
The computations are completely analogous to those of Proposition \ref{pr5}, and we have the
ordering  
\begin{equation*}
\text{\rm FPR}(n,\theta,R_3)\le  \text{\rm FPR}(n,\theta,R_2)\le \text{\rm FPR}(n,\theta,R_1)\ ,
\end{equation*}
 as with the positive rates.
 
\section{Main results}\label{res}
In this section we will use the concepts  
from ROC analysis introduced in Section \ref{roc} as a numeric criterion
to compare the relative goodness of decision rules. Theorem \ref{teor1} establishes the preference
order of the three rules considered, under the criterion of greater area of the triangle,
where it is seen that $R_1$ is uniformly the best. This is still true 
when the false negatives are more penalised than the false positives (Corollary \ref{cor:wsmall}).
If false positives are deemed worse, the situation is more complex and interesting; it will
be covered by Theorem \ref{teor12}.

\begin{defi}
A rule $R$ is  \emph{AOT-better} than a rule $R'$ if and only if,  for all $n$ odd and $\theta> \frac{1}{2}$,
\begin{equation}
 \text{\rm AOT}(n,\theta,R) \geq  \text{\rm AOT}(n,\theta,R')\ ,   \label{crit}
\end{equation}
 and the inequality is strict for some value of $n$ or $\theta$.
\end{defi}

Under our model assumptions, it is now shown that rule $R_1$ is   AOT-better than $R_2$ and that  
$R_2$ is AOT-better than $R_3$: 

\begin{teor}\label{teor1}
  Under the model conditions (C1)--(C4), for all $n\ge 3$ odd and for all $\theta> \frac{1}{2}$, 
\begin{equation*}
\text{\rm AOT}(n,\theta,R_3) \le  \text{\rm AOT}(n,\theta,R_2) 
<  \text{\rm AOT}(n,\theta,R_1)
\end{equation*}  
and the first inequality is strict for $n\ge 7$.
\end{teor}  

\begin{proof}
We start with the second inequality, which is equivalent,
by (\ref{aot2}), to 
$\text{TPR}(n,\theta,R_2)-\text{FPR}(n,\theta,R_2) < \text{TPR}(n,\theta,R_1)-\text{FPR}(n,\theta,R_1)$,
and therefore to 
\begin{equation}\label{eq:FPRTPRineq}
\text{FPR}(n,\theta,R_1)-\text{FPR}(n,\theta,R_2) < \text{TPR}(n,\theta,R_1)-\text{TPR}(n,\theta,R_2)\ .
\end{equation}
In other words, it is equivalent to say that the increase in true positives when changing from $R_2$ to $R_1$
more than compensates the increase in false positives.

Using formulas (\ref{R1tpr}), (\ref{R2tpr}),  (\ref{R1fpr}) and (\ref{R2fpr}),  it is clearly enough
to prove that, 
for all $0\le i\le n$, for all $0\le j\le n-i$, and for all $ m-i\wedge j+1\le k\le m-\lfloor \frac{i\wedge j}{2}\rfloor$, 
\begin{equation}
\theta^{n-i+j}(1-\theta)^{n+i-j}<\theta^{i+j+2k}(1-\theta)^{2n-i-j-2k} \nonumber
\end{equation}
or, equivalently, that
$$
(1-\theta)^{2i+2k-n}<\theta^{2i+2k-n}\ ,   
$$
which can be easily checked taking into account that $0<1-\theta<\theta$ and $2i+2k-n>0$: 
Indeed,  $n=2m+1$ and $k\geq m-i\wedge j +1$ 
imply that $2i+2k-n \geq 2i+2(m-i\wedge j +1)-(2m+1)=2i-2i\wedge j+1\geq 1$. Thus, the second inequality is proved. 

A similar argument proves the first one, using (\ref{R2tpr}), (\ref{R3tpr}), 
(\ref{R2fpr}) and (\ref{R3fpr}), and
noticing that  
$k\geq m-i\wedge j +1$ also holds.
the inequality is strict, unless $n\le 5$, in which case  $R_2$ and $R_3$ are equal 
(see Proposition \ref{pr2}).
\end{proof}

For the weighted area of the
triangle defined by formula (\ref{WAOT}),
and weights $w<\frac{1}{2}$ 
(that means, when false negatives
are considered more harmful than false positives), the relations between $R_1$, $R_2$ and $R_3$
are the same as with AOT (case $w=\frac{1}{2}$), as stated in the next  Corollary \ref{cor:wsmall}.
However, for $w>\frac{1}{2}$, none of the rules gives a greater WAOT than another,
uniformly in $n\ge 3$ and $\frac{1}{2}<\theta<1$;   
this will be precisely stated and proved in 
Theorem \ref{teor12}.

\begin{cor}\label{cor:wsmall}
  Under the model conditions (C1)--(C4), for all $n\ge 3$ odd, and for all $\theta>\frac{1}{2}$ 
  and $w<\frac{1}{2}$, 
\begin{equation*}
\text{\rm WAOT}_w(n,\theta,R_3) \le  \text{\rm WAOT}_w(n,\theta,R_2) 
<  \text{\rm WAOT}_w(n,\theta,R_1)\ ,
\end{equation*}  
and the first inequality is strict for $n\ge 7$.
\end{cor}

\begin{proof}
The proof of Theorem \ref{teor1} is based in checking the inequality (\ref{eq:FPRTPRineq}).
The analogue for the weighted area is 
$w\big(\text{FPR}(n,\theta,R_1)-\text{FPR}(n,\theta,R_2)\big) < 
(1-w)\big(\text{TPR}(n,\theta,R_1)-\text{TPR}(n,\theta,R_2)\big)$, which trivially follows from (\ref{eq:FPRTPRineq})
when $w<\frac{1}{2}$. 
\end{proof}


The case $w>\frac{1}{2}$ is different. The relation between the WAOT of $R_1$ and $R_2$ 
is the same
if the competence $\theta$ stands above a certain threshold 
$C(w)$, with $\frac{1}{2}<C(w)<w$, and similarly with $R_2$ and $R_3$.    
This is made 
precise in the next theorem.


\begin{teor}\label{teor12}
Fix $n\ge 3$. For every weight $\frac{1}{2}<w<1$, there exists $C_1(w)$, smaller than $w$ (except that
$C_1(w)=w$ if $n=3$), such that
\begin{align*}
\theta> C_1(w) &\Rightarrow \text{\rm WAOT}_w(n,\theta,R_1) > \text{\rm WAOT}_w(n,\theta,R_2)\ .
\end{align*}
Fix $n\ge 7$. For every weight $\frac{1}{2}<w<1$, there exists $C_2(w)$, smaller than $w$, such that
\begin{align*}
\theta> C_2(w) &\Rightarrow \text{\rm WAOT}_w(n,\theta,R_2) > \text{\rm WAOT}_w(n,\theta,R_3) \ .
\end{align*}
\end{teor}  
 
The theorem will follow from the following two lemmas.
In  the proofs, we only treat in detail the claims relating $R_1$ and $R_2$, those relating $R_2$ and $R_3$
being analogous, with only slight changes that will be noted.
The condition $n ≥ 7$ in the second case is needed
since otherwise $R_2$ and $R_3$ coincide, as we have seen in Section \ref{3rules}.

The first lemma is interesting in itself in that it establishes a dichotomy when we look
at $\theta$ as fixed and let $w$ vary. The second lemma states some technical properties of the 
functions introduced in the first.
Refer to Figure \ref{fig:D(theta)} for a graphical clue of the situation presented in theorem
and lemmas.

\begin{lema}\label{existD}
Fix $n\ge 3$. For every competence $\frac{1}{2}<\theta<1$, there exists a constant $D_1(\theta)$, 
greater than $\theta$ (except that $D_1(\theta)=\theta$ if $n=3$), such that
\begin{align*}
w< D_1(\theta) &\Rightarrow \text{\rm WAOT}_w(n,\theta,R_1) > \text{\rm WAOT}_w(n,\theta,R_2) \\
w> D_1(\theta) &\Rightarrow \text{\rm WAOT}_w(n,\theta,R_1) < \text{\rm WAOT}_w(n,\theta,R_2)\ .
\end{align*}
Fix $n\ge 7$. For every competence $\frac{1}{2}<\theta<1$, there exists a constant $D_2(\theta)$, greater than $\theta$, such that
\begin{align*}
w< D_2(\theta) &\Rightarrow \text{\rm WAOT}_w(n,\theta,R_2) > \text{\rm WAOT}_w(n,\theta,R_3) \\
w> D_2(\theta) &\Rightarrow \text{\rm WAOT}_w(n,\theta,R_2) < \text{\rm WAOT}_w(n,\theta,R_3) \ .
\end{align*} 
\end{lema}

\begin{lema}\label{Dlim}
For $i=1,2$,  the functions $\theta \mapsto D_i(\theta)$ in Proposition \ref{existD} 
satisfy:
\begin{enumerate}
\item For any $\frac{1}{2}<\theta<1$, we have also $\frac{1}{2} < D_i(\theta) < 1$.
\item $D_i$ is continuous in the interval $(\frac{1}{2},1)$.
\item   $\lim_{\theta\searrow \frac{1}{2}} D_i(\theta)=\frac{1}{2}$ and  
  $\lim_{\theta\nearrow 1} D_i(\theta) = 1 $. 
\end{enumerate}
\end{lema}

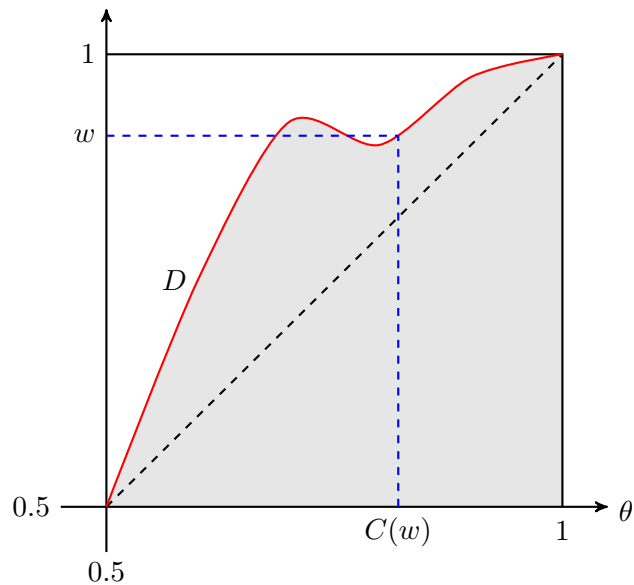
\begin{figure}[h] 
  \centering 
  \begin{tikzpicture}[
  thick,
  >=stealth',
  scale=6,
  ]
  \coordinate (O) at (0,0);
     \fill [gray!20] 
          (0, 0)
          -- plot[smooth] coordinates {(0,0) (0.20,0.50) (0.40,0.85) (0.60,0.80) (0.80,0.95) (1,1)}
          -- (1, 0) 
          -- cycle;
  
  \draw[->] (-0.1,0) -- (1.1,0) ;
  \draw[->] (0,-0.1) -- (0,1.1);
  \draw[-] (1,0) -- (1,1) -- (0,1);
  \draw [-, dashed] ( 0,0) -- ( 1, 1); 
  \node [right] at (1.10,-0.01) {$\theta$};
  \node [left] at (0.00, 0.82) {$w$};
  \node [left] at (0.20,0.50) {$D$};
  \node [below] at (0.64,0) {$C(w)$};
  \node [left] at (-0.1, 0) {$0.5$};
  \node [below] at (0, -0.1) {$0.5$};
  \node [below] at (1.00,-0.01) {$1$};
  \node [left] at (-0.00,1) {$1$};
  \draw[red] plot[smooth] coordinates {(0,0) (0.20,0.50) (0.40,0.85) (0.60,0.80) (0.80,0.95) (1,1)};
  \draw [-, thick, blue, dashed] ( 0,0.82) -- ( 0.64, 0.82);  
  \draw [-, thick, blue, dashed] ( 0.64,0) -- ( 0.64, 0.82);          
  \end{tikzpicture}   
  \caption{An example graph (in red) of a function $D(\theta)$ with the properties of Lemmas \ref{existD} 
  and \ref{Dlim}.
  For competence values $\theta$ greater than $C(w)$, 
  rule $R_1$ is better than rule $R_2$ in the sense of the weighted area of the triangle.
  \label{fig:D(theta)}}
\end{figure}

\emph{Proof of Lemma \ref{existD}.}

We prove the first part of the lemma, the other one being analogous, and write simply $D$ 
instead of $D_1$.

Fix $n\ge 3$,  
and let $\frac{1}{2}<\theta<1$. 
From the definitions of WAOT, FPR, FNR and TPR, the inequality 
$\text{WAOT}_w(n,\theta,R_2)<\text{WAOT}_w(n,\theta,R_1)$
is equivalent to 
\begin{equation*}
w\big(\text{FPR}(n,\theta,R_1)-\text{FPR}(n,\theta,R_2)\big)
<
(1-w)\big(\text{TPR}(n,\theta,R_1)-\text{TPR}(n,\theta,R_2)\big)\ ,
\end{equation*}
or 
\begin{equation}\label{eq:main_ineq2}
(1-w)\big(\text{TPR}(n,\theta,R_1)-\text{TPR}(n,\theta,R_2)\big)
-w\big(\text{FPR}(n,\theta,R_1)-\text{FPR}(n,\theta,R_2)\big)
>0\ .
\end{equation}
We know that inequality (\ref{eq:main_ineq2}) is true for $w\le \frac{1}{2}$, 
for all $n\ge 3$ and all $\theta>\frac{1}{2}$,
by Theorem \ref{teor1} and Corollary \ref{cor:wsmall}, whereas for $w=1$ is 
manifestly false, since clearly $\text{FPR}(n,\theta,R_1)-\text{FPR}(n,\theta,R_2)>0$,
from (\ref{R1fpr})-(\ref{R2fpr}).

Since the left-hand side is a linear function of $w$, there must be a unique point 
$D(\theta)$ such that the equality holds in (\ref{eq:main_ineq2}), and such that for every $w<D(\theta)$ rule
$R_1$ yields a greater weighted area than $R_2$, whereas for every $w>D(\theta)$ it is the 
other way round.

Let us now prove than $D(\theta)>\theta$: Consider the particular weight $w=\theta$.
We only need to check that (\ref{eq:main_ineq2}) is satisfied for this $w$. Indeed,
using (\ref{R1tpr}--\ref{R2tpr}) and (\ref{R1fpr}--\ref{R2fpr}),
inequality (\ref{eq:main_ineq2}) is equivalent, for this particular value, to
\begin{equation}\label{eq:positivity}
\sum_{i=0}^n\ \sum_{j=0}^{n-i}\ \sum_{k=m-i\wedge j+1}^{(m-\lfloor\frac{i\wedge j}{2}\rfloor)\wedge (n-i-j)}
{n\choose i,j,k,\ell}
\Big[
\theta^{i+j+2k}(1-\theta)^{1+2n-i-j-2k} - \theta^{1+n-i+j}(1-\theta)^{n+i-j} 
\Big]
>0 \ .
\end{equation}
The expression in square brackets is non-negative if and only if 
\begin{equation*}
\Big(\frac{\theta}{1-\theta}\Big)^{2(i+k)-(n+1)}\ge 1\ .
\end{equation*}  
Taking into account that we are assuming $\theta>\frac{1}{2}$, and that $2(i+k)\ge n+1$ 
in this range of indices, the sum
(\ref{eq:positivity}) is non-negative. 
One can easily check that there is at least
one positive term in the sum for $n>3$.
By the linearity in $w$, and the fact that $w=\theta$ satisfies (\ref{eq:main_ineq2}), 
it is clear that the critical point $D(\theta)$ is greater than $\theta$.
In the special case $n=3$, there is only one term in the sum and it is equal to zero, hence 
$D(\theta)=\theta$.

The claim on $R_2$ and $R_3$ can be proved in the same way, with the only difference that
the index $k$ in the sum (\ref{eq:positivity}) ranges from $m-\lfloor \frac{i\wedge j}{2}\rfloor +1$ 
to $m\wedge(n-i-j)$. If $n=7$, the only term in the sum (\ref{eq:positivity}) is already positive,
and there is no need to consider this limit case separately.

\emph{Proof of  Lemma \ref{Dlim}.} 

Again, we prove the claim for $D_1$ and call it simply $D$. The proofs for $D_2$ are similar,
with the differences pointed out at the end.
 
Fix $n\geq 3$. The point $D(\theta)$ is easily computed: $D(\theta)=A(\theta)/(A(\theta)+B(\theta))$, with
\begin{equation*}
 A(\theta):=\text{\rm TPR}(n,\theta,R_1)-\text{\rm TPR}(n,\theta,R_2)
  \ ,\quad 
  B(\theta):=\text{\rm FPR}(n,\theta,R_1)-\text{\rm FPR}(n,\theta,R_2)
  \ .
\end{equation*} 
Using (\ref{R1tpr}--\ref{R2tpr}) and (\ref{R1fpr}--\ref{R2fpr}), their explicit expressions are:
\begin{eqnarray*} 
A(\theta)&=&\sum_{i=0}^n\ \sum_{j=0}^{n-i}\ \sum_{k=m-i\wedge j+1}^{(m-\lfloor\frac{i\wedge j}{2}\rfloor)\wedge (n-i-j)}
{n\choose i,j,k,\ell} \theta^{i+j+2k}(1-\theta)^{2n-i-j-2k}  
 \\
B(\theta)&=&\sum_{i=0}^n\ \sum_{j=0}^{n-i}\ \sum_{k=m-i\wedge j+1}^{(m-\lfloor\frac{i\wedge j}{2}\rfloor)\wedge (n-i-j)}
{n\choose i,j,k,\ell} \theta^{n-i+j}(1-\theta)^{n+i-j}\ . 
\end{eqnarray*}
First, notice that the function $D(\theta)$ is the quotient of two polynomials  in $\theta$  
that never vanish because both the true and the false positive rates are greater for rule 
$R_1$  than for $R_2$, for all $\theta>\frac{1}{2}$ and $n\ge 3$ (see again (\ref{R1tpr})-(\ref{R2tpr}) and 
(\ref{R1fpr})-(\ref{R2fpr})).
Hence, it is clear that $D$ is continuous and less than 1 in its domain. Moreover, Lemma \ref{existD} 
ensures that $\frac{1}{2}<\theta\le D(\theta)$.
In particular, $\frac12<D(\theta)<1$. For $\theta=\frac{1}{2}$, both $A$ and $B$ are well defined and 
$A(\frac{1}{2})=B(\frac{1}{2})$, giving $\lim_{\theta\searrow \frac{1}{2}} D(\theta)=\frac{1}{2}$.  

To complete the proof of the lemma, we must see that  $\lim_{\theta\nearrow 1} D_i(\theta) = 1 $. 
This is equivalent to see that $\lim_{\theta\nearrow 1}\psi(\theta)=0$, where  
$\psi(\theta) :=\frac{B(\theta)}{A(\theta)}$. 

For simplicity, let us  
denote by $\Omega$ the valid range of the indices in the triple sums above, 
\begin{align*}
 \Omega:= \{\omega=(i,j,k,\ell): \ &0\leq i\leq n,\ 0\leq j\leq n-i, \\
   &m-i\wedge j+1\leq k\leq(m-\lfloor\tfrac{i\wedge j}{2}\rfloor)\wedge (n-i-j),\  
   \ell=n-i-j-k \}
   \ ,
\end{align*}
and also   
\begin{equation*} 
 \alpha:={n\choose i,j,k,\ell}\,\quad\text{and}\quad  
 \gamma:=\frac{\theta}{1-\theta}\ .
\end{equation*}    
 Then we can write 
\begin{equation*} 
A(\gamma) = \Big(\frac{1}{1+\gamma}\Big)^{2n} \displaystyle\sum_{\omega\in\Omega} \alpha  \gamma^{i+j+2k}\ ,    
\quad 
B(\gamma)= \Big(\frac{1}{1+\gamma}\Big)^{2n} \displaystyle\sum_{\omega\in\Omega} \alpha \gamma^{n-i+j}   
\label{eq:ABgamma} 
\end{equation*}    
and
\begin{equation} 
\psi(\gamma)=\frac{\displaystyle\sum_{\omega\in\Omega} \alpha \gamma^{n-i+j}}
                  {\displaystyle\sum_{\omega\in\Omega} \ \alpha  \gamma^{i+j+2k}}
                  \ .
\label{eq:psi}
\end{equation}    
The parameter $\gamma$ is an increasing function of  $\theta$, mapping $(\frac{1}{2},1)$ into  $(1,\infty)$.     
Therefore, it is enough to show that $\psi(\gamma)$ converges to 0 as $\gamma\nearrow \infty$. 
To see this, consider the highest powers of the polynomials in (\ref{eq:psi}). 
Let us check that 
\begin{equation}
\displaystyle\max_{\omega\in\Omega} (n-i+j)=n+m-1  
\quad 
\mbox{and} 
\quad 
\displaystyle\max_{\omega\in\Omega} (i+j+2k)=n+m  
\ ;
\label{eq:maxs}
\end{equation}
then clearly 
$\displaystyle\lim_{\gamma\nearrow \infty} \psi(\gamma)=0$ 
and the proof will be complete.
 
Indeed, notice that $\omega^*=(1,m,m,0)$ belongs to $\Omega$ and achieves both values 
$n-m-1$ and $n+m$. Moreover, for any 
$\omega=(i,j,k,\ell)\in \Omega$, we have  
$i+j+2k\leq n+k\leq n+m$.
For the other bound, consider two cases: If $j\le i$, the inequality $n-i+j\le n+m-1$ is trivial;
for $j>i$, 
we have on $\Omega$ the condition $m-i+1\le n-i-j$, equivalent to $m\ge j$. Since we have also $i\ge 1$
on $\Omega$, we deduce $m\ge -i+j+1$. Adding $n-1$ to both sides, we get again the desired 
inequality.
Thus, (\ref{eq:maxs}) is checked and this finishes the proof for $D_1$.  

The claims for $D_2$ can be proved in the same way with the difference that, in the correspondingly
defined set $\Omega$, index $k$
ranges from $m-\lfloor \frac{i\wedge j}{2}\rfloor +1$ 
to $m\wedge(n-i-j)$, and that the maximum value of $(n-i-j)$ on $\Omega$ 
is $n+m-3$; the maximum of $(i+j+2k)$ is again $n+m$, 
and both maxima are achieved at the point $w=(2,m-1,m,0)\in\Omega$. 
The remaining arguments of the proof are identical.

\emph{Proof of  Theorem \ref{teor12}.} 

Denote, as before, by $D$ and $C$ de functions $D_1$ and $C_1$, respectively. The proof of the second
part of the theorem is identical, with $D_2$ and $C_2$.

From Lemma \ref{existD} and Lemma \ref{Dlim}, we know that $D(\theta)$  defines a 
continuous function that maps $[\frac{1}{2},1]$ onto $[\frac{1}{2},1]$ (extending it by continuity at the endpoints), and 
that the curve $\big(\theta, D(\theta)\big)$ lies above the diagonal, as depicted in Figure \ref{fig:D(theta)}. 
The shaded region below the curve is
the set of pairs $(\theta, w)$ for which 
$\text{\rm WAOT}_w(n,\theta,R_1)>\text{\rm WAOT}_w(n,\theta,R_2)$, and the inequality
is reversed above the curve. 

Regardless whether $D$ is and increasing function or not, and since it is a quotient of polynomials,
for any fixed $w$ in $(\frac{1}{2},1)$ there is a largest value $C(w)\in (\frac{1}{2},1)$ that solves for $\theta$
the equation $D(\theta)=w$. And it is clear from the figure that for all $\theta>C(w)$, all points 
of the segment $\{(\theta,w):\ \theta>C(w)\}$ lie below the curve,
hence $\text{\rm WAOT}_w(n,\theta,R_1)>\text{\rm WAOT}_w(n,\theta,R_2)$ for these values.  

Finally, since $D(\theta)>\theta$ for all $\theta$ if $n>3$, we have in particular $w=D(C(w))>C(w)$.
For $n=3$, $D(\theta)=\theta$, and we get $w=C(w)$. This finishes the proof.
\qed


\section{Examples}\label{sec:examples}
In this section we illustrate the above theory with some numeric computations and figures. 

It is clear from the previous sections that none of the rules considered is best
for all pairs $(\theta,w)$ of competence and weight. In fact, no two rules $R_i$ and $R_j$ are
comparable, uniformly in $\theta$ and $w$ (except for $R_2$ and $R_3$ when $n\le 5$, because 
they coincide). Indeed, Table \ref{tab:C-WAOT} shows the different possible orders under the  
WAOT criterion for three fixed competence values and for varying $w\in(0,1)$, and committee size $n=11$.

\begin{table}[h] \centering
{\footnotesize
    {\extrarowheight 2pt
\begin{tabular}{|c|c|c|c|}
  \hline
   & \centering $\theta=0.6$ & $\theta=0.75$  & $\theta=0.90$  \\
  \hline
 $R_3< R_2< R_1$  &  ~~~~~~~$0  <w <  0.6930$  & ~~~~~~~$0  <w < 0.8722$ & ~~~~~~~$0  <w <  0.9576$\\
 $R_3< R_1 < R_2$  &  $0.6930< w <0.7184$  & $0.8722< w <0.9154$  & $0.9576< w <0.9847$ \\
 $R_1< R_3< R_2$  &  $0.7184< w <0.8215$  & $0.9154< w <0.9820$  & $0.9847< w <0.9995$ \\
 $R_1< R_2< R_3$  &  $0.8215< w < 1$~~~~~~~  & $0.9820< w < 1$~~~~~~~  & 
 $0.9995< w < 1$~~~~~~~ \\
  \hline
\end{tabular} 
}    }
\caption{For each $\theta$, intervals of weights $w$ where each order of rules holds true. Te committee
size is $n=11$. 
The symbol $<$ in the first column means here the relation ``worse than'' with respect to the
WAOT criterion.
\label{tab:C-WAOT}
}
\end{table} 

In Table \ref{tres1}, again with committee size $n=11$, the values of TPR, FPR, and AOT are 
computed to four decimal places for a large range of competence values, using (\ref{R1tpr}--\ref{R3tpr}) 
and (\ref{R1fpr}--\ref{R3fpr}). The last column is the value of 
WAOT for a fixed
weight $w=0.75$; 
in other words, false positives penalises the performance measure three times more than false negatives. In column five, 
they penalise in the same proportion. 

Under the AOT criterion, $R_1$ (IbyI) is always better than $R_2$ (PbyP), and $R_2$ is better than $R_3$ (CbyC),
as Theorem \ref{teor1} claims. The numbers in the table  
give an idea of the extent of the difference, suggesting that $R_2$ and $R_3$ are closer together
than $R_1$ and $R_2$. We also see that for very high values of 
$\theta$ the rules get closer and approach fast to the perfect value 0.5. 

\begin{table}[ht]
\centering
\begingroup\footnotesize
\begin{tabular}{rrrrrrr}
 $n$ & $\theta$ & rule & TPR & FPR & AOT & WAOT \\ 
  \hline
11 & 0.55 & 1 & 0.4008 & 0.2323 & 0.0843 & 0.1760 \\ 
  11 & 0.55 & 2 & 0.1372 & 0.0589 & 0.0391 & 0.2401 \\ 
  11 & 0.55 & 3 & 0.0811 & 0.0327 & 0.0242 & 0.2457 \\ 
   \hline
11 & 0.60 & 1 & 0.5678 & 0.1857 & 0.1910 & 0.2526 \\ 
  11 & 0.60 & 2 & 0.2569 & 0.0480 & 0.1044 & 0.2782 \\ 
  11 & 0.60 & 3 & 0.1661 & 0.0283 & 0.0689 & 0.2703 \\ 
   \hline
11 & 0.65 & 1 & 0.7247 & 0.1266 & 0.2991 & 0.3363 \\ 
  11 & 0.65 & 2 & 0.4197 & 0.0340 & 0.1928 & 0.3294 \\ 
  11 & 0.65 & 3 & 0.2984 & 0.0219 & 0.1382 & 0.3081 \\ 
   \hline
11 & 0.70 & 1 & 0.8497 & 0.0721 & 0.3888 & 0.4083 \\ 
  11 & 0.70 & 2 & 0.6050 & 0.0207 & 0.2922 & 0.3857 \\ 
  11 & 0.70 & 3 & 0.4729 & 0.0148 & 0.2291 & 0.3571 \\ 
   \hline
11 & 0.75 & 1 & 0.9325 & 0.0331 & 0.4497 & 0.4583 \\ 
  11 & 0.75 & 2 & 0.7779 & 0.0105 & 0.3837 & 0.4366 \\ 
  11 & 0.75 & 3 & 0.6649 & 0.0084 & 0.3282 & 0.4099 \\ 
   \hline
11 & 0.80 & 1 & 0.9768 & 0.0115 & 0.4827 & 0.4856 \\ 
  11 & 0.80 & 2 & 0.9051 & 0.0042 & 0.4505 & 0.4731 \\ 
  11 & 0.80 & 3 & 0.8339 & 0.0037 & 0.4151 & 0.4557 \\ 
   \hline
11 & 0.85 & 1 & 0.9947 & 0.0026 & 0.4960 & 0.4967 \\ 
  11 & 0.85 & 2 & 0.9735 & 0.0012 & 0.4862 & 0.4925 \\ 
  11 & 0.85 & 3 & 0.9446 & 0.0011 & 0.4717 & 0.4853 \\ 
   \hline
11 & 0.90 & 1 & 0.9994 & 0.0003 & 0.4996 & 0.4996 \\ 
  11 & 0.90 & 2 & 0.9965 & 0.0002 & 0.4982 & 0.4990 \\ 
  11 & 0.90 & 3 & 0.9910 & 0.0002 & 0.4954 & 0.4976 \\ 
   \hline
11 & 0.95 & 1 & 1.0000 & 0.0000 & 0.5000 & 0.5000 \\ 
  11 & 0.95 & 2 & 0.9999 & 0.0000 & 0.5000 & 0.5000 \\ 
  11 & 0.95 & 3 & 0.9997 & 0.0000 & 0.4999 & 0.4999 \\ 
   \hline
\end{tabular}
\endgroup
\caption{Comparison of decision rules 
$R_1$, $R_2$ and $R_3$ in the ROC map for a fixed number of voters $n=11$ 
and different competence levels $\theta$. 
Notice that $R_1$ has the largest values of AOT for any fixed $\theta$, 
that is, the sequence of ordered rules from best to worst is always 1, 2, 3. 
are recorded in the last column.} 
\label{tres1}
\end{table}

For the WAOT criterion, with $w=0.75$, we see that for low competence values of the jury, it is better
to use rules $R_2$ or $R_3$. At some point, the order of AOT is re-established and preserved till the
end of the table. From Table \ref{tab:C-WAOT} we could deduce that this happens somewhere between $0.60$ and 
$0.75$, and here we see that it must be for some value between $0.60$ and $0.65$. Of course the 
exact value can be computed, and it turns out to be 0.6374.


A simple illustration of the evolution of the AOT for the three rules we are considering, 
for several committee sizes,
can be seen in Figure \ref{fig:Theta-AOT}. For $n=3,7,11$ the AOT value is drawn against $\theta$,
with red ($R_1$), blue ($R_2$) and green ($R_3$) curves.  
Notice that the largest absolute differences in AOT take place around the middle values of the competence
 range. That means, for $0.6 \lesssim\theta\lesssim 0.8$, say, it is when
 the election of the decision rule is most critical.
 
The analogous Figure \ref{fig:Theta-WAOT} shows the same curves, in the case $w=0.75$. 
Rule $R_1$ is the worst in the lower end of $\theta$ values, and the best in the upper end.
Rule $R_3$ does the opposite.

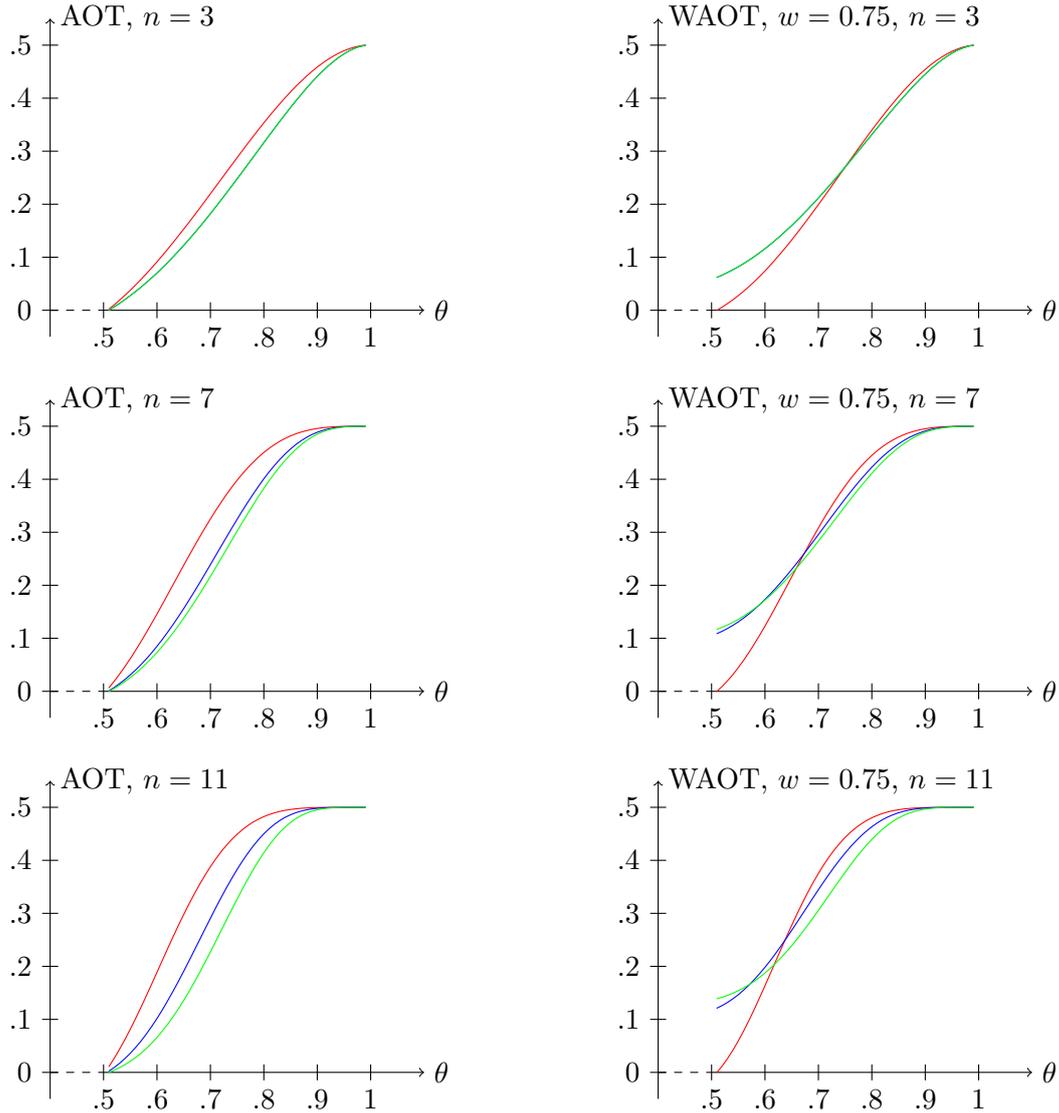
\begin{figure}[t!]
\centering
\begin{subfigure}[t]{.45\textwidth}
  \centering
\begin{tikzpicture}[x=1pt,y=1pt, scale=2]
\definecolor{fillColor}{RGB}{255,255,255}
\path[use as bounding box,fill=fillColor,fill opacity=0.00] (40,40) rectangle (110,100);
       
   \draw [->] ( 50, 50) -- ( 110, 50) node[right] {$\theta$};
   \draw [-, dashed] ( 40, 50) -- ( 50, 50);
   \draw [->] ( 40, 45) -- ( 40, 105) node[above, right] {AOT, $n=3$};
   \draw [-] ( 50, 48.5) node [below] {$.5$} -- ( 50, 51.5); 
   \draw [-] ( 60, 48.5) node [below] {$.6$} -- ( 60, 51.5); 
   \draw [-] ( 70, 48.5) node [below] {$.7$} -- ( 70, 51.5); 
   \draw [-] ( 80, 48.5) node [below] {$.8$} -- ( 80, 51.5); 
   \draw [-] ( 90, 48.5) node [below] {$.9$} -- ( 90, 51.5); 
   \draw [-] ( 100, 48.5) node [below] {$1$} -- ( 100, 51.5);
   \draw [-] ( 38.5, 50) node [left] {$0$} -- ( 41.5, 50);
   \draw [-] ( 38.5, 60) node [left] {$.1$} -- ( 41.5, 60);
   \draw [-] ( 38.5, 70) node [left] {$.2$} -- ( 41.5, 70);
   \draw [-] ( 38.5, 80) node [left] {$.3$} -- ( 41.5, 80);
   \draw [-] ( 38.5, 90) node [left] {$.4$} -- ( 41.5, 90);
   \draw [-] ( 38.5, 100) node [left] {$.5$} -- ( 41.5, 100);   
\begin{scope}[shift={(51-40.88547, 50-98.7666)}, 
              xscale=48/(115.44-53.10), 
              yscale=50/(94.08-62.46)
]

\definecolor{drawColor}{RGB}{255,0,0}

\path[draw=drawColor,line width= 0.4pt,line join=round,line cap=round] ( 53.10, 62.59) --
	( 54.40, 63.11) --
	( 55.69, 63.66) --
	( 56.99, 64.24) --
	( 58.29, 64.84) --
	( 59.59, 65.47) --
	( 60.89, 66.13) --
	( 62.19, 66.81) --
	( 63.49, 67.51) --
	( 64.79, 68.23) --
	( 66.09, 68.98) --
	( 67.38, 69.74) --
	( 68.68, 70.52) --
	( 69.98, 71.32) --
	( 71.28, 72.13) --
	( 72.58, 72.95) --
	( 73.88, 73.79) --
	( 75.18, 74.63) --
	( 76.48, 75.49) --
	( 77.78, 76.35) --
	( 79.07, 77.21) --
	( 80.37, 78.08) --
	( 81.67, 78.94) --
	( 82.97, 79.80) --
	( 84.27, 80.66) --
	( 85.57, 81.51) --
	( 86.87, 82.36) --
	( 88.17, 83.19) --
	( 89.47, 84.01) --
	( 90.76, 84.81) --
	( 92.06, 85.59) --
	( 93.36, 86.36) --
	( 94.66, 87.10) --
	( 95.96, 87.82) --
	( 97.26, 88.51) --
	( 98.56, 89.17) --
	( 99.86, 89.79) --
	(101.16, 90.39) --
	(102.45, 90.95) --
	(103.75, 91.47) --
	(105.05, 91.94) --
	(106.35, 92.38) --
	(107.65, 92.77) --
	(108.95, 93.12) --
	(110.25, 93.41) --
	(111.55, 93.66) --
	(112.85, 93.85) --
	(114.14, 93.99) --
	(115.44, 94.08);
\definecolor{drawColor}{RGB}{0,0,255}

\path[draw=drawColor,line width= 0.4pt,line join=round,line cap=round] ( 53.10, 62.46) --
	( 54.40, 62.86) --
	( 55.69, 63.28) --
	( 56.99, 63.73) --
	( 58.29, 64.20) --
	( 59.59, 64.69) --
	( 60.89, 65.21) --
	( 62.19, 65.75) --
	( 63.49, 66.32) --
	( 64.79, 66.90) --
	( 66.09, 67.52) --
	( 67.38, 68.15) --
	( 68.68, 68.81) --
	( 69.98, 69.49) --
	( 71.28, 70.19) --
	( 72.58, 70.91) --
	( 73.88, 71.65) --
	( 75.18, 72.41) --
	( 76.48, 73.18) --
	( 77.78, 73.98) --
	( 79.07, 74.78) --
	( 80.37, 75.60) --
	( 81.67, 76.44) --
	( 82.97, 77.28) --
	( 84.27, 78.13) --
	( 85.57, 78.99) --
	( 86.87, 79.85) --
	( 88.17, 80.72) --
	( 89.47, 81.58) --
	( 90.76, 82.45) --
	( 92.06, 83.31) --
	( 93.36, 84.16) --
	( 94.66, 85.01) --
	( 95.96, 85.83) --
	( 97.26, 86.65) --
	( 98.56, 87.44) --
	( 99.86, 88.21) --
	(101.16, 88.96) --
	(102.45, 89.67) --
	(103.75, 90.35) --
	(105.05, 90.98) --
	(106.35, 91.58) --
	(107.65, 92.12) --
	(108.95, 92.61) --
	(110.25, 93.04) --
	(111.55, 93.41) --
	(112.85, 93.70) --
	(114.14, 93.92) --
	(115.44, 94.06);
\definecolor{drawColor}{RGB}{0,255,0}

\path[draw=drawColor,line width= 0.4pt,line join=round,line cap=round] ( 53.10, 62.46) --
	( 54.40, 62.86) --
	( 55.69, 63.28) --
	( 56.99, 63.73) --
	( 58.29, 64.20) --
	( 59.59, 64.69) --
	( 60.89, 65.21) --
	( 62.19, 65.75) --
	( 63.49, 66.32) --
	( 64.79, 66.90) --
	( 66.09, 67.52) --
	( 67.38, 68.15) --
	( 68.68, 68.81) --
	( 69.98, 69.49) --
	( 71.28, 70.19) --
	( 72.58, 70.91) --
	( 73.88, 71.65) --
	( 75.18, 72.41) --
	( 76.48, 73.18) --
	( 77.78, 73.98) --
	( 79.07, 74.78) --
	( 80.37, 75.60) --
	( 81.67, 76.44) --
	( 82.97, 77.28) --
	( 84.27, 78.13) --
	( 85.57, 78.99) --
	( 86.87, 79.85) --
	( 88.17, 80.72) --
	( 89.47, 81.58) --
	( 90.76, 82.45) --
	( 92.06, 83.31) --
	( 93.36, 84.16) --
	( 94.66, 85.01) --
	( 95.96, 85.83) --
	( 97.26, 86.65) --
	( 98.56, 87.44) --
	( 99.86, 88.21) --
	(101.16, 88.96) --
	(102.45, 89.67) --
	(103.75, 90.35) --
	(105.05, 90.98) --
	(106.35, 91.58) --
	(107.65, 92.12) --
	(108.95, 92.61) --
	(110.25, 93.04) --
	(111.55, 93.41) --
	(112.85, 93.70) --
	(114.14, 93.92) --
	(115.44, 94.06);
\end{scope}
\end{tikzpicture}

\vspace{0.8cm}
\begin{tikzpicture}[x=1pt,y=1pt, scale=2]
\definecolor{fillColor}{RGB}{255,255,255}
\path[use as bounding box,fill=fillColor,fill opacity=0.00] (40,40) rectangle (110,100);
       
   \draw [->] ( 50, 50) -- ( 110, 50) node[right] {$\theta$};
   \draw [-, dashed] ( 40, 50) -- ( 50, 50);
   \draw [->] ( 40, 45) -- ( 40, 105) node[above, right] {AOT, $n=7$};
   \draw [-] ( 50, 48.5) node [below] {$.5$} -- ( 50, 51.5); 
   \draw [-] ( 60, 48.5) node [below] {$.6$} -- ( 60, 51.5); 
   \draw [-] ( 70, 48.5) node [below] {$.7$} -- ( 70, 51.5); 
   \draw [-] ( 80, 48.5) node [below] {$.8$} -- ( 80, 51.5); 
   \draw [-] ( 90, 48.5) node [below] {$.9$} -- ( 90, 51.5); 
   \draw [-] ( 100, 48.5) node [below] {$1$} -- ( 100, 51.5);
   \draw [-] ( 38.5, 50) node [left] {$0$} -- ( 41.5, 50);
   \draw [-] ( 38.5, 60) node [left] {$.1$} -- ( 41.5, 60);
   \draw [-] ( 38.5, 70) node [left] {$.2$} -- ( 41.5, 70);
   \draw [-] ( 38.5, 80) node [left] {$.3$} -- ( 41.5, 80);
   \draw [-] ( 38.5, 90) node [left] {$.4$} -- ( 41.5, 90);
   \draw [-] ( 38.5, 100) node [left] {$.5$} -- ( 41.5, 100);   
\begin{scope}[shift={(51-40.88547, 50-98.7666)}, 
              xscale=48/(115.44-53.10), 
              yscale=50/(94.08-62.46)
]

\definecolor{drawColor}{RGB}{255,0,0}

\path[draw=drawColor,line width= 0.4pt,line join=round,line cap=round] ( 53.10, 62.88) --
	( 54.40, 63.67) --
	( 55.69, 64.51) --
	( 56.99, 65.41) --
	( 58.29, 66.36) --
	( 59.59, 67.35) --
	( 60.89, 68.39) --
	( 62.19, 69.46) --
	( 63.49, 70.56) --
	( 64.79, 71.68) --
	( 66.09, 72.83) --
	( 67.38, 73.99) --
	( 68.68, 75.16) --
	( 69.98, 76.32) --
	( 71.28, 77.49) --
	( 72.58, 78.64) --
	( 73.88, 79.77) --
	( 75.18, 80.88) --
	( 76.48, 81.97) --
	( 77.78, 83.02) --
	( 79.07, 84.03) --
	( 80.37, 85.00) --
	( 81.67, 85.93) --
	( 82.97, 86.81) --
	( 84.27, 87.64) --
	( 85.57, 88.41) --
	( 86.87, 89.13) --
	( 88.17, 89.79) --
	( 89.47, 90.40) --
	( 90.76, 90.95) --
	( 92.06, 91.45) --
	( 93.36, 91.89) --
	( 94.66, 92.29) --
	( 95.96, 92.63) --
	( 97.26, 92.93) --
	( 98.56, 93.18) --
	( 99.86, 93.39) --
	(101.16, 93.56) --
	(102.45, 93.70) --
	(103.75, 93.81) --
	(105.05, 93.90) --
	(106.35, 93.96) --
	(107.65, 94.01) --
	(108.95, 94.04) --
	(110.25, 94.06) --
	(111.55, 94.07) --
	(112.85, 94.07) --
	(114.14, 94.08) --
	(115.44, 94.08);
\definecolor{drawColor}{RGB}{0,0,255}

\path[draw=drawColor,line width= 0.4pt,line join=round,line cap=round] ( 53.10, 62.52) --
	( 54.40, 62.92) --
	( 55.69, 63.38) --
	( 56.99, 63.87) --
	( 58.29, 64.42) --
	( 59.59, 65.00) --
	( 60.89, 65.64) --
	( 62.19, 66.32) --
	( 63.49, 67.05) --
	( 64.79, 67.83) --
	( 66.09, 68.65) --
	( 67.38, 69.51) --
	( 68.68, 70.41) --
	( 69.98, 71.35) --
	( 71.28, 72.32) --
	( 72.58, 73.33) --
	( 73.88, 74.36) --
	( 75.18, 75.41) --
	( 76.48, 76.49) --
	( 77.78, 77.57) --
	( 79.07, 78.66) --
	( 80.37, 79.75) --
	( 81.67, 80.84) --
	( 82.97, 81.92) --
	( 84.27, 82.98) --
	( 85.57, 84.02) --
	( 86.87, 85.02) --
	( 88.17, 85.99) --
	( 89.47, 86.92) --
	( 90.76, 87.81) --
	( 92.06, 88.64) --
	( 93.36, 89.41) --
	( 94.66, 90.13) --
	( 95.96, 90.78) --
	( 97.26, 91.37) --
	( 98.56, 91.90) --
	( 99.86, 92.36) --
	(101.16, 92.75) --
	(102.45, 93.09) --
	(103.75, 93.36) --
	(105.05, 93.58) --
	(106.35, 93.75) --
	(107.65, 93.87) --
	(108.95, 93.96) --
	(110.25, 94.02) --
	(111.55, 94.05) --
	(112.85, 94.07) --
	(114.14, 94.07) --
	(115.44, 94.08);
\definecolor{drawColor}{RGB}{0,255,0}

\path[draw=drawColor,line width= 0.4pt,line join=round,line cap=round] ( 53.10, 62.46) --
	( 54.40, 62.81) --
	( 55.69, 63.20) --
	( 56.99, 63.63) --
	( 58.29, 64.10) --
	( 59.59, 64.61) --
	( 60.89, 65.17) --
	( 62.19, 65.77) --
	( 63.49, 66.41) --
	( 64.79, 67.10) --
	( 66.09, 67.83) --
	( 67.38, 68.61) --
	( 68.68, 69.42) --
	( 69.98, 70.28) --
	( 71.28, 71.18) --
	( 72.58, 72.12) --
	( 73.88, 73.08) --
	( 75.18, 74.08) --
	( 76.48, 75.11) --
	( 77.78, 76.16) --
	( 79.07, 77.22) --
	( 80.37, 78.30) --
	( 81.67, 79.39) --
	( 82.97, 80.48) --
	( 84.27, 81.56) --
	( 85.57, 82.63) --
	( 86.87, 83.68) --
	( 88.17, 84.71) --
	( 89.47, 85.71) --
	( 90.76, 86.67) --
	( 92.06, 87.59) --
	( 93.36, 88.46) --
	( 94.66, 89.27) --
	( 95.96, 90.02) --
	( 97.26, 90.71) --
	( 98.56, 91.33) --
	( 99.86, 91.88) --
	(101.16, 92.37) --
	(102.45, 92.78) --
	(103.75, 93.13) --
	(105.05, 93.41) --
	(106.35, 93.63) --
	(107.65, 93.79) --
	(108.95, 93.91) --
	(110.25, 93.99) --
	(111.55, 94.04) --
	(112.85, 94.06) --
	(114.14, 94.07) --
	(115.44, 94.08);
\end{scope}
\end{tikzpicture}

\vspace{0.8cm}
\begin{tikzpicture}[x=1pt,y=1pt, scale=2]
\definecolor{fillColor}{RGB}{255,255,255}
\path[use as bounding box,fill=fillColor,fill opacity=0.00] (40,40) rectangle (110,100);
       
   \draw [->] ( 50, 50) -- ( 110, 50) node[right] {$\theta$};
   \draw [-, dashed] ( 40, 50) -- ( 50, 50);
   \draw [->] ( 40, 45) -- ( 40, 105) node[above, right] {AOT, $n=11$};
   \draw [-] ( 50, 48.5) node [below] {$.5$} -- ( 50, 51.5); 
   \draw [-] ( 60, 48.5) node [below] {$.6$} -- ( 60, 51.5); 
   \draw [-] ( 70, 48.5) node [below] {$.7$} -- ( 70, 51.5); 
   \draw [-] ( 80, 48.5) node [below] {$.8$} -- ( 80, 51.5); 
   \draw [-] ( 90, 48.5) node [below] {$.9$} -- ( 90, 51.5); 
   \draw [-] ( 100, 48.5) node [below] {$1$} -- ( 100, 51.5);
   \draw [-] ( 38.5, 50) node [left] {$0$} -- ( 41.5, 50);
   \draw [-] ( 38.5, 60) node [left] {$.1$} -- ( 41.5, 60);
   \draw [-] ( 38.5, 70) node [left] {$.2$} -- ( 41.5, 70);
   \draw [-] ( 38.5, 80) node [left] {$.3$} -- ( 41.5, 80);
   \draw [-] ( 38.5, 90) node [left] {$.4$} -- ( 41.5, 90);
   \draw [-] ( 38.5, 100) node [left] {$.5$} -- ( 41.5, 100);   
\begin{scope}[shift={(51-40.88547, 50-98.7666)}, 
              xscale=48/(115.44-53.10), 
              yscale=50/(94.08-62.46)
]

\definecolor{drawColor}{RGB}{255,0,0}

\path[draw=drawColor,line width= 0.4pt,line join=round,line cap=round] ( 53.10, 63.15) --
	( 54.40, 64.15) --
	( 55.69, 65.23) --
	( 56.99, 66.38) --
	( 58.29, 67.61) --
	( 59.59, 68.89) --
	( 60.89, 70.22) --
	( 62.19, 71.59) --
	( 63.49, 72.99) --
	( 64.79, 74.40) --
	( 66.09, 75.82) --
	( 67.38, 77.22) --
	( 68.68, 78.61) --
	( 69.98, 79.97) --
	( 71.28, 81.28) --
	( 72.58, 82.55) --
	( 73.88, 83.76) --
	( 75.18, 84.91) --
	( 76.48, 85.99) --
	( 77.78, 86.99) --
	( 79.07, 87.92) --
	( 80.37, 88.78) --
	( 81.67, 89.55) --
	( 82.97, 90.25) --
	( 84.27, 90.87) --
	( 85.57, 91.42) --
	( 86.87, 91.90) --
	( 88.17, 92.32) --
	( 89.47, 92.67) --
	( 90.76, 92.97) --
	( 92.06, 93.22) --
	( 93.36, 93.42) --
	( 94.66, 93.59) --
	( 95.96, 93.72) --
	( 97.26, 93.82) --
	( 98.56, 93.90) --
	( 99.86, 93.96) --
	(101.16, 94.00) --
	(102.45, 94.03) --
	(103.75, 94.05) --
	(105.05, 94.06) --
	(106.35, 94.07) --
	(107.65, 94.07) --
	(108.95, 94.07) --
	(110.25, 94.08) --
	(111.55, 94.08) --
	(112.85, 94.08) --
	(114.14, 94.08) --
	(115.44, 94.08);
\definecolor{drawColor}{RGB}{0,0,255}

\path[draw=drawColor,line width= 0.4pt,line join=round,line cap=round] ( 53.10, 62.62) --
	( 54.40, 63.05) --
	( 55.69, 63.54) --
	( 56.99, 64.11) --
	( 58.29, 64.73) --
	( 59.59, 65.43) --
	( 60.89, 66.19) --
	( 62.19, 67.03) --
	( 63.49, 67.93) --
	( 64.79, 68.89) --
	( 66.09, 69.91) --
	( 67.38, 71.00) --
	( 68.68, 72.13) --
	( 69.98, 73.30) --
	( 71.28, 74.52) --
	( 72.58, 75.76) --
	( 73.88, 77.03) --
	( 75.18, 78.30) --
	( 76.48, 79.58) --
	( 77.78, 80.84) --
	( 79.07, 82.09) --
	( 80.37, 83.30) --
	( 81.67, 84.48) --
	( 82.97, 85.60) --
	( 84.27, 86.67) --
	( 85.57, 87.67) --
	( 86.87, 88.60) --
	( 88.17, 89.46) --
	( 89.47, 90.23) --
	( 90.76, 90.92) --
	( 92.06, 91.53) --
	( 93.36, 92.06) --
	( 94.66, 92.51) --
	( 95.96, 92.89) --
	( 97.26, 93.19) --
	( 98.56, 93.44) --
	( 99.86, 93.63) --
	(101.16, 93.78) --
	(102.45, 93.88) --
	(103.75, 93.96) --
	(105.05, 94.01) --
	(106.35, 94.04) --
	(107.65, 94.06) --
	(108.95, 94.07) --
	(110.25, 94.07) --
	(111.55, 94.07) --
	(112.85, 94.08) --
	(114.14, 94.08) --
	(115.44, 94.08);
\definecolor{drawColor}{RGB}{0,255,0}

\path[draw=drawColor,line width= 0.4pt,line join=round,line cap=round] ( 53.10, 62.46) --
	( 54.40, 62.73) --
	( 55.69, 63.03) --
	( 56.99, 63.38) --
	( 58.29, 63.78) --
	( 59.59, 64.23) --
	( 60.89, 64.74) --
	( 62.19, 65.31) --
	( 63.49, 65.94) --
	( 64.79, 66.63) --
	( 66.09, 67.38) --
	( 67.38, 68.20) --
	( 68.68, 69.09) --
	( 69.98, 70.03) --
	( 71.28, 71.04) --
	( 72.58, 72.11) --
	( 73.88, 73.22) --
	( 75.18, 74.39) --
	( 76.48, 75.59) --
	( 77.78, 76.82) --
	( 79.07, 78.08) --
	( 80.37, 79.35) --
	( 81.67, 80.63) --
	( 82.97, 81.89) --
	( 84.27, 83.14) --
	( 85.57, 84.35) --
	( 86.87, 85.52) --
	( 88.17, 86.64) --
	( 89.47, 87.69) --
	( 90.76, 88.67) --
	( 92.06, 89.57) --
	( 93.36, 90.38) --
	( 94.66, 91.10) --
	( 95.96, 91.73) --
	( 97.26, 92.28) --
	( 98.56, 92.73) --
	( 99.86, 93.10) --
	(101.16, 93.39) --
	(102.45, 93.62) --
	(103.75, 93.78) --
	(105.05, 93.90) --
	(106.35, 93.98) --
	(107.65, 94.03) --
	(108.95, 94.05) --
	(110.25, 94.07) --
	(111.55, 94.07) --
	(112.85, 94.08) --
	(114.14, 94.08) --
	(115.44, 94.08);
\end{scope}
\end{tikzpicture}
\caption{AOT. Rule $R_1$ is uniformly best in $\theta$.
    \label{fig:Theta-AOT}}
\end{subfigure}%
\hspace{0.05\textwidth}%
\begin{subfigure}[t]{.45\textwidth}
  \centering
\begin{tikzpicture}[x=1pt,y=1pt, scale=2]
\definecolor{fillColor}{RGB}{255,255,255}
\path[use as bounding box,fill=fillColor,fill opacity=0.00] (40,40) rectangle (110,100);
       
   \draw [->] ( 50, 50) -- ( 110, 50) node[right] {$\theta$};
   \draw [-, dashed] ( 40, 50) -- ( 50, 50);
   \draw [->] ( 40, 45) -- ( 40, 105) node[above, right] {WAOT, $w=0.75$, $n=3$};
   \draw [-] ( 50, 48.5) node [below] {$.5$} -- ( 50, 51.5); 
   \draw [-] ( 60, 48.5) node [below] {$.6$} -- ( 60, 51.5); 
   \draw [-] ( 70, 48.5) node [below] {$.7$} -- ( 70, 51.5); 
   \draw [-] ( 80, 48.5) node [below] {$.8$} -- ( 80, 51.5); 
   \draw [-] ( 90, 48.5) node [below] {$.9$} -- ( 90, 51.5); 
   \draw [-] ( 100, 48.5) node [below] {$1$} -- ( 100, 51.5);
   \draw [-] ( 38.5, 50) node [left] {$0$} -- ( 41.5, 50);
   \draw [-] ( 38.5, 60) node [left] {$.1$} -- ( 41.5, 60);
   \draw [-] ( 38.5, 70) node [left] {$.2$} -- ( 41.5, 70);
   \draw [-] ( 38.5, 80) node [left] {$.3$} -- ( 41.5, 80);
   \draw [-] ( 38.5, 90) node [left] {$.4$} -- ( 41.5, 90);
   \draw [-] ( 38.5, 100) node [left] {$.5$} -- ( 41.5, 100);   
\begin{scope}[shift={(51-40.88547, 50-98.7666)}, 
              xscale=48/(115.44-53.10), 
              yscale=50/(94.08-62.46)
]

\definecolor{drawColor}{RGB}{255,0,0}

\path[draw=drawColor,line width= 0.4pt,line join=round,line cap=round] ( 53.10, 62.46) --
	( 54.40, 62.84) --
	( 55.69, 63.26) --
	( 56.99, 63.71) --
	( 58.29, 64.20) --
	( 59.59, 64.72) --
	( 60.89, 65.28) --
	( 62.19, 65.87) --
	( 63.49, 66.49) --
	( 64.79, 67.15) --
	( 66.09, 67.83) --
	( 67.38, 68.55) --
	( 68.68, 69.29) --
	( 69.98, 70.05) --
	( 71.28, 70.84) --
	( 72.58, 71.65) --
	( 73.88, 72.48) --
	( 75.18, 73.33) --
	( 76.48, 74.19) --
	( 77.78, 75.06) --
	( 79.07, 75.94) --
	( 80.37, 76.83) --
	( 81.67, 77.73) --
	( 82.97, 78.63) --
	( 84.27, 79.53) --
	( 85.57, 80.43) --
	( 86.87, 81.32) --
	( 88.17, 82.21) --
	( 89.47, 83.08) --
	( 90.76, 83.94) --
	( 92.06, 84.79) --
	( 93.36, 85.61) --
	( 94.66, 86.42) --
	( 95.96, 87.19) --
	( 97.26, 87.95) --
	( 98.56, 88.67) --
	( 99.86, 89.35) --
	(101.16, 90.01) --
	(102.45, 90.62) --
	(103.75, 91.19) --
	(105.05, 91.72) --
	(106.35, 92.20) --
	(107.65, 92.63) --
	(108.95, 93.01) --
	(110.25, 93.34) --
	(111.55, 93.61) --
	(112.85, 93.83) --
	(114.14, 93.98) --
	(115.44, 94.08);
\definecolor{drawColor}{RGB}{0,0,255}

\path[draw=drawColor,line width= 0.4pt,line join=round,line cap=round] ( 53.10, 66.38) --
	( 54.40, 66.65) --
	( 55.69, 66.96) --
	( 56.99, 67.29) --
	( 58.29, 67.64) --
	( 59.59, 68.02) --
	( 60.89, 68.43) --
	( 62.19, 68.86) --
	( 63.49, 69.32) --
	( 64.79, 69.80) --
	( 66.09, 70.31) --
	( 67.38, 70.84) --
	( 68.68, 71.39) --
	( 69.98, 71.97) --
	( 71.28, 72.56) --
	( 72.58, 73.18) --
	( 73.88, 73.82) --
	( 75.18, 74.48) --
	( 76.48, 75.16) --
	( 77.78, 75.85) --
	( 79.07, 76.56) --
	( 80.37, 77.28) --
	( 81.67, 78.02) --
	( 82.97, 78.77) --
	( 84.27, 79.53) --
	( 85.57, 80.30) --
	( 86.87, 81.08) --
	( 88.17, 81.86) --
	( 89.47, 82.64) --
	( 90.76, 83.42) --
	( 92.06, 84.20) --
	( 93.36, 84.97) --
	( 94.66, 85.74) --
	( 95.96, 86.50) --
	( 97.26, 87.24) --
	( 98.56, 87.96) --
	( 99.86, 88.67) --
	(101.16, 89.35) --
	(102.45, 90.01) --
	(103.75, 90.63) --
	(105.05, 91.22) --
	(106.35, 91.76) --
	(107.65, 92.27) --
	(108.95, 92.72) --
	(110.25, 93.12) --
	(111.55, 93.46) --
	(112.85, 93.73) --
	(114.14, 93.94) --
	(115.44, 94.06);
\definecolor{drawColor}{RGB}{0,255,0}

\path[draw=drawColor,line width= 0.4pt,line join=round,line cap=round] ( 53.10, 66.38) --
	( 54.40, 66.65) --
	( 55.69, 66.96) --
	( 56.99, 67.29) --
	( 58.29, 67.64) --
	( 59.59, 68.02) --
	( 60.89, 68.43) --
	( 62.19, 68.86) --
	( 63.49, 69.32) --
	( 64.79, 69.80) --
	( 66.09, 70.31) --
	( 67.38, 70.84) --
	( 68.68, 71.39) --
	( 69.98, 71.97) --
	( 71.28, 72.56) --
	( 72.58, 73.18) --
	( 73.88, 73.82) --
	( 75.18, 74.48) --
	( 76.48, 75.16) --
	( 77.78, 75.85) --
	( 79.07, 76.56) --
	( 80.37, 77.28) --
	( 81.67, 78.02) --
	( 82.97, 78.77) --
	( 84.27, 79.53) --
	( 85.57, 80.30) --
	( 86.87, 81.08) --
	( 88.17, 81.86) --
	( 89.47, 82.64) --
	( 90.76, 83.42) --
	( 92.06, 84.20) --
	( 93.36, 84.97) --
	( 94.66, 85.74) --
	( 95.96, 86.50) --
	( 97.26, 87.24) --
	( 98.56, 87.96) --
	( 99.86, 88.67) --
	(101.16, 89.35) --
	(102.45, 90.01) --
	(103.75, 90.63) --
	(105.05, 91.22) --
	(106.35, 91.76) --
	(107.65, 92.27) --
	(108.95, 92.72) --
	(110.25, 93.12) --
	(111.55, 93.46) --
	(112.85, 93.73) --
	(114.14, 93.94) --
	(115.44, 94.06);
\end{scope}
\end{tikzpicture}

\vspace{0.8cm}
\begin{tikzpicture}[x=1pt,y=1pt, scale=2]
\definecolor{fillColor}{RGB}{255,255,255}
\path[use as bounding box,fill=fillColor,fill opacity=0.00] (40,40) rectangle (110,100);
       
   \draw [->] ( 50, 50) -- ( 110, 50) node[right] {$\theta$};
   \draw [-, dashed] ( 40, 50) -- ( 50, 50);
   \draw [->] ( 40, 45) -- ( 40, 105) node[above, right] {WAOT, $w=0.75$, $n=7$};
   \draw [-] ( 50, 48.5) node [below] {$.5$} -- ( 50, 51.5); 
   \draw [-] ( 60, 48.5) node [below] {$.6$} -- ( 60, 51.5); 
   \draw [-] ( 70, 48.5) node [below] {$.7$} -- ( 70, 51.5); 
   \draw [-] ( 80, 48.5) node [below] {$.8$} -- ( 80, 51.5); 
   \draw [-] ( 90, 48.5) node [below] {$.9$} -- ( 90, 51.5); 
   \draw [-] ( 100, 48.5) node [below] {$1$} -- ( 100, 51.5);
   \draw [-] ( 38.5, 50) node [left] {$0$} -- ( 41.5, 50);
   \draw [-] ( 38.5, 60) node [left] {$.1$} -- ( 41.5, 60);
   \draw [-] ( 38.5, 70) node [left] {$.2$} -- ( 41.5, 70);
   \draw [-] ( 38.5, 80) node [left] {$.3$} -- ( 41.5, 80);
   \draw [-] ( 38.5, 90) node [left] {$.4$} -- ( 41.5, 90);
   \draw [-] ( 38.5, 100) node [left] {$.5$} -- ( 41.5, 100);   
\begin{scope}[shift={(51-40.88547, 50-98.7666)}, 
              xscale=48/(115.44-53.10), 
              yscale=50/(94.08-62.46)
]

\definecolor{drawColor}{RGB}{255,0,0}

\path[draw=drawColor,line width= 0.4pt,line join=round,line cap=round] ( 53.10, 62.46) --
	( 54.40, 63.05) --
	( 55.69, 63.72) --
	( 56.99, 64.46) --
	( 58.29, 65.28) --
	( 59.59, 66.16) --
	( 60.89, 67.10) --
	( 62.19, 68.10) --
	( 63.49, 69.15) --
	( 64.79, 70.24) --
	( 66.09, 71.37) --
	( 67.38, 72.53) --
	( 68.68, 73.71) --
	( 69.98, 74.90) --
	( 71.28, 76.10) --
	( 72.58, 77.30) --
	( 73.88, 78.49) --
	( 75.18, 79.66) --
	( 76.48, 80.82) --
	( 77.78, 81.94) --
	( 79.07, 83.03) --
	( 80.37, 84.08) --
	( 81.67, 85.09) --
	( 82.97, 86.04) --
	( 84.27, 86.95) --
	( 85.57, 87.80) --
	( 86.87, 88.59) --
	( 88.17, 89.32) --
	( 89.47, 89.99) --
	( 90.76, 90.60) --
	( 92.06, 91.15) --
	( 93.36, 91.65) --
	( 94.66, 92.08) --
	( 95.96, 92.46) --
	( 97.26, 92.79) --
	( 98.56, 93.07) --
	( 99.86, 93.31) --
	(101.16, 93.50) --
	(102.45, 93.66) --
	(103.75, 93.78) --
	(105.05, 93.88) --
	(106.35, 93.95) --
	(107.65, 94.00) --
	(108.95, 94.03) --
	(110.25, 94.05) --
	(111.55, 94.07) --
	(112.85, 94.07) --
	(114.14, 94.08) --
	(115.44, 94.08);
\definecolor{drawColor}{RGB}{0,0,255}

\path[draw=drawColor,line width= 0.4pt,line join=round,line cap=round] ( 53.10, 69.34) --
	( 54.40, 69.63) --
	( 55.69, 69.96) --
	( 56.99, 70.33) --
	( 58.29, 70.74) --
	( 59.59, 71.20) --
	( 60.89, 71.69) --
	( 62.19, 72.23) --
	( 63.49, 72.80) --
	( 64.79, 73.42) --
	( 66.09, 74.07) --
	( 67.38, 74.75) --
	( 68.68, 75.47) --
	( 69.98, 76.22) --
	( 71.28, 76.99) --
	( 72.58, 77.79) --
	( 73.88, 78.61) --
	( 75.18, 79.45) --
	( 76.48, 80.30) --
	( 77.78, 81.16) --
	( 79.07, 82.02) --
	( 80.37, 82.88) --
	( 81.67, 83.74) --
	( 82.97, 84.58) --
	( 84.27, 85.42) --
	( 85.57, 86.23) --
	( 86.87, 87.02) --
	( 88.17, 87.78) --
	( 89.47, 88.50) --
	( 90.76, 89.19) --
	( 92.06, 89.84) --
	( 93.36, 90.45) --
	( 94.66, 91.00) --
	( 95.96, 91.51) --
	( 97.26, 91.97) --
	( 98.56, 92.38) --
	( 99.86, 92.74) --
	(101.16, 93.04) --
	(102.45, 93.30) --
	(103.75, 93.52) --
	(105.05, 93.69) --
	(106.35, 93.82) --
	(107.65, 93.92) --
	(108.95, 93.98) --
	(110.25, 94.03) --
	(111.55, 94.06) --
	(112.85, 94.07) --
	(114.14, 94.07) --
	(115.44, 94.08);
\definecolor{drawColor}{RGB}{0,255,0}

\path[draw=drawColor,line width= 0.4pt,line join=round,line cap=round] ( 53.10, 69.85) --
	( 54.40, 70.10) --
	( 55.69, 70.38) --
	( 56.99, 70.69) --
	( 58.29, 71.04) --
	( 59.59, 71.43) --
	( 60.89, 71.85) --
	( 62.19, 72.31) --
	( 63.49, 72.81) --
	( 64.79, 73.34) --
	( 66.09, 73.90) --
	( 67.38, 74.51) --
	( 68.68, 75.14) --
	( 69.98, 75.81) --
	( 71.28, 76.51) --
	( 72.58, 77.23) --
	( 73.88, 77.99) --
	( 75.18, 78.76) --
	( 76.48, 79.56) --
	( 77.78, 80.37) --
	( 79.07, 81.19) --
	( 80.37, 82.03) --
	( 81.67, 82.87) --
	( 82.97, 83.70) --
	( 84.27, 84.54) --
	( 85.57, 85.36) --
	( 86.87, 86.17) --
	( 88.17, 86.96) --
	( 89.47, 87.72) --
	( 90.76, 88.46) --
	( 92.06, 89.16) --
	( 93.36, 89.82) --
	( 94.66, 90.44) --
	( 95.96, 91.01) --
	( 97.26, 91.53) --
	( 98.56, 92.00) --
	( 99.86, 92.42) --
	(101.16, 92.79) --
	(102.45, 93.10) --
	(103.75, 93.36) --
	(105.05, 93.57) --
	(106.35, 93.74) --
	(107.65, 93.86) --
	(108.95, 93.95) --
	(110.25, 94.01) --
	(111.55, 94.05) --
	(112.85, 94.07) --
	(114.14, 94.07) --
	(115.44, 94.08);
\end{scope}
\end{tikzpicture}

\vspace{0.8cm}
\begin{tikzpicture}[x=1pt,y=1pt, scale=2]
\definecolor{fillColor}{RGB}{255,255,255}
\path[use as bounding box,fill=fillColor,fill opacity=0.00] (40,40) rectangle (110,100);
       
   \draw [->] ( 50, 50) -- ( 110, 50) node[right] {$\theta$};
   \draw [-, dashed] ( 40, 50) -- ( 50, 50);
   \draw [->] ( 40, 45) -- ( 40, 105) node[above, right] {WAOT, $w=0.75$, $n=11$};
   \draw [-] ( 50, 48.5) node [below] {$.5$} -- ( 50, 51.5); 
   \draw [-] ( 60, 48.5) node [below] {$.6$} -- ( 60, 51.5); 
   \draw [-] ( 70, 48.5) node [below] {$.7$} -- ( 70, 51.5); 
   \draw [-] ( 80, 48.5) node [below] {$.8$} -- ( 80, 51.5); 
   \draw [-] ( 90, 48.5) node [below] {$.9$} -- ( 90, 51.5); 
   \draw [-] ( 100, 48.5) node [below] {$1$} -- ( 100, 51.5);
   \draw [-] ( 38.5, 50) node [left] {$0$} -- ( 41.5, 50);
   \draw [-] ( 38.5, 60) node [left] {$.1$} -- ( 41.5, 60);
   \draw [-] ( 38.5, 70) node [left] {$.2$} -- ( 41.5, 70);
   \draw [-] ( 38.5, 80) node [left] {$.3$} -- ( 41.5, 80);
   \draw [-] ( 38.5, 90) node [left] {$.4$} -- ( 41.5, 90);
   \draw [-] ( 38.5, 100) node [left] {$.5$} -- ( 41.5, 100);   
\begin{scope}[shift={(51-40.88547, 50-98.7666)}, 
              xscale=48/(115.44-53.10), 
              yscale=50/(94.08-62.46)
]

\definecolor{drawColor}{RGB}{255,0,0}

\path[draw=drawColor,line width= 0.4pt,line join=round,line cap=round] ( 53.10, 62.46) --
	( 54.40, 63.23) --
	( 55.69, 64.12) --
	( 56.99, 65.11) --
	( 58.29, 66.21) --
	( 59.59, 67.39) --
	( 60.89, 68.66) --
	( 62.19, 69.99) --
	( 63.49, 71.38) --
	( 64.79, 72.80) --
	( 66.09, 74.25) --
	( 67.38, 75.70) --
	( 68.68, 77.15) --
	( 69.98, 78.59) --
	( 71.28, 79.99) --
	( 72.58, 81.35) --
	( 73.88, 82.66) --
	( 75.18, 83.91) --
	( 76.48, 85.09) --
	( 77.78, 86.19) --
	( 79.07, 87.22) --
	( 80.37, 88.16) --
	( 81.67, 89.02) --
	( 82.97, 89.79) --
	( 84.27, 90.49) --
	( 85.57, 91.10) --
	( 86.87, 91.64) --
	( 88.17, 92.10) --
	( 89.47, 92.50) --
	( 90.76, 92.83) --
	( 92.06, 93.11) --
	( 93.36, 93.34) --
	( 94.66, 93.53) --
	( 95.96, 93.68) --
	( 97.26, 93.79) --
	( 98.56, 93.88) --
	( 99.86, 93.94) --
	(101.16, 93.99) --
	(102.45, 94.02) --
	(103.75, 94.04) --
	(105.05, 94.06) --
	(106.35, 94.07) --
	(107.65, 94.07) --
	(108.95, 94.07) --
	(110.25, 94.07) --
	(111.55, 94.08) --
	(112.85, 94.08) --
	(114.14, 94.08) --
	(115.44, 94.08);
\definecolor{drawColor}{RGB}{0,0,255}

\path[draw=drawColor,line width= 0.4pt,line join=round,line cap=round] ( 53.10, 70.12) --
	( 54.40, 70.44) --
	( 55.69, 70.81) --
	( 56.99, 71.24) --
	( 58.29, 71.72) --
	( 59.59, 72.27) --
	( 60.89, 72.87) --
	( 62.19, 73.53) --
	( 63.49, 74.24) --
	( 64.79, 75.00) --
	( 66.09, 75.80) --
	( 67.38, 76.65) --
	( 68.68, 77.54) --
	( 69.98, 78.46) --
	( 71.28, 79.40) --
	( 72.58, 80.36) --
	( 73.88, 81.34) --
	( 75.18, 82.31) --
	( 76.48, 83.29) --
	( 77.78, 84.25) --
	( 79.07, 85.19) --
	( 80.37, 86.10) --
	( 81.67, 86.99) --
	( 82.97, 87.83) --
	( 84.27, 88.62) --
	( 85.57, 89.37) --
	( 86.87, 90.05) --
	( 88.17, 90.69) --
	( 89.47, 91.25) --
	( 90.76, 91.76) --
	( 92.06, 92.21) --
	( 93.36, 92.60) --
	( 94.66, 92.93) --
	( 95.96, 93.20) --
	( 97.26, 93.43) --
	( 98.56, 93.61) --
	( 99.86, 93.75) --
	(101.16, 93.86) --
	(102.45, 93.94) --
	(103.75, 93.99) --
	(105.05, 94.03) --
	(106.35, 94.05) --
	(107.65, 94.06) --
	(108.95, 94.07) --
	(110.25, 94.07) --
	(111.55, 94.08) --
	(112.85, 94.08) --
	(114.14, 94.08) --
	(115.44, 94.08);
\definecolor{drawColor}{RGB}{0,255,0}

\path[draw=drawColor,line width= 0.4pt,line join=round,line cap=round] ( 53.10, 71.25) --
	( 54.40, 71.43) --
	( 55.69, 71.65) --
	( 56.99, 71.91) --
	( 58.29, 72.20) --
	( 59.59, 72.54) --
	( 60.89, 72.92) --
	( 62.19, 73.34) --
	( 63.49, 73.80) --
	( 64.79, 74.32) --
	( 66.09, 74.88) --
	( 67.38, 75.48) --
	( 68.68, 76.14) --
	( 69.98, 76.83) --
	( 71.28, 77.57) --
	( 72.58, 78.35) --
	( 73.88, 79.17) --
	( 75.18, 80.02) --
	( 76.48, 80.89) --
	( 77.78, 81.78) --
	( 79.07, 82.69) --
	( 80.37, 83.61) --
	( 81.67, 84.53) --
	( 82.97, 85.43) --
	( 84.27, 86.33) --
	( 85.57, 87.19) --
	( 86.87, 88.03) --
	( 88.17, 88.82) --
	( 89.47, 89.57) --
	( 90.76, 90.26) --
	( 92.06, 90.90) --
	( 93.36, 91.47) --
	( 94.66, 91.98) --
	( 95.96, 92.43) --
	( 97.26, 92.81) --
	( 98.56, 93.13) --
	( 99.86, 93.39) --
	(101.16, 93.60) --
	(102.45, 93.75) --
	(103.75, 93.87) --
	(105.05, 93.95) --
	(106.35, 94.01) --
	(107.65, 94.04) --
	(108.95, 94.06) --
	(110.25, 94.07) --
	(111.55, 94.07) --
	(112.85, 94.08) --
	(114.14, 94.08) --
	(115.44, 94.08);
\end{scope}
\end{tikzpicture}
\caption{WAOT. For any $w>1/2$, there is no a
         uniformly best or worst rule.
    \label{fig:Theta-WAOT}}
\end{subfigure}
\caption{AOT and WAOT for each of the rules $R_1$ (red), $R_2$ (blue), $R_3$ (green) as 
  a function of $\theta$, for several committee sizes $n$}
\end{figure}

%
%
%
%
%
%
 
\FloatBarrier

Combining the committee 
sizes  
$n=3,7,11$ and the competence values $\theta=0.60,0.75,0.90$, we draw the
triangles in ROC space of the three decision rules,  
in Figure \ref{rocs}.
In this picture, it can be observed that the area of the triangle  determined by rule $R_1$ (in red) 
is larger than the  area determined by  $R_2$ (in blue), which in its turn is larger than the area 
determined by rule $R_3$ (in green) for $n>5$, 
 and that the triangles of $R_2$ and $R_3$ coincide for $n=3$ and $n=5$.


\begin{figure}[!ht]
  \begin{center}
\begin{knitrout}
\definecolor{shadecolor}{rgb}{0.969, 0.969, 0.969}\color{fgcolor}
\includegraphics[width=\maxwidth]{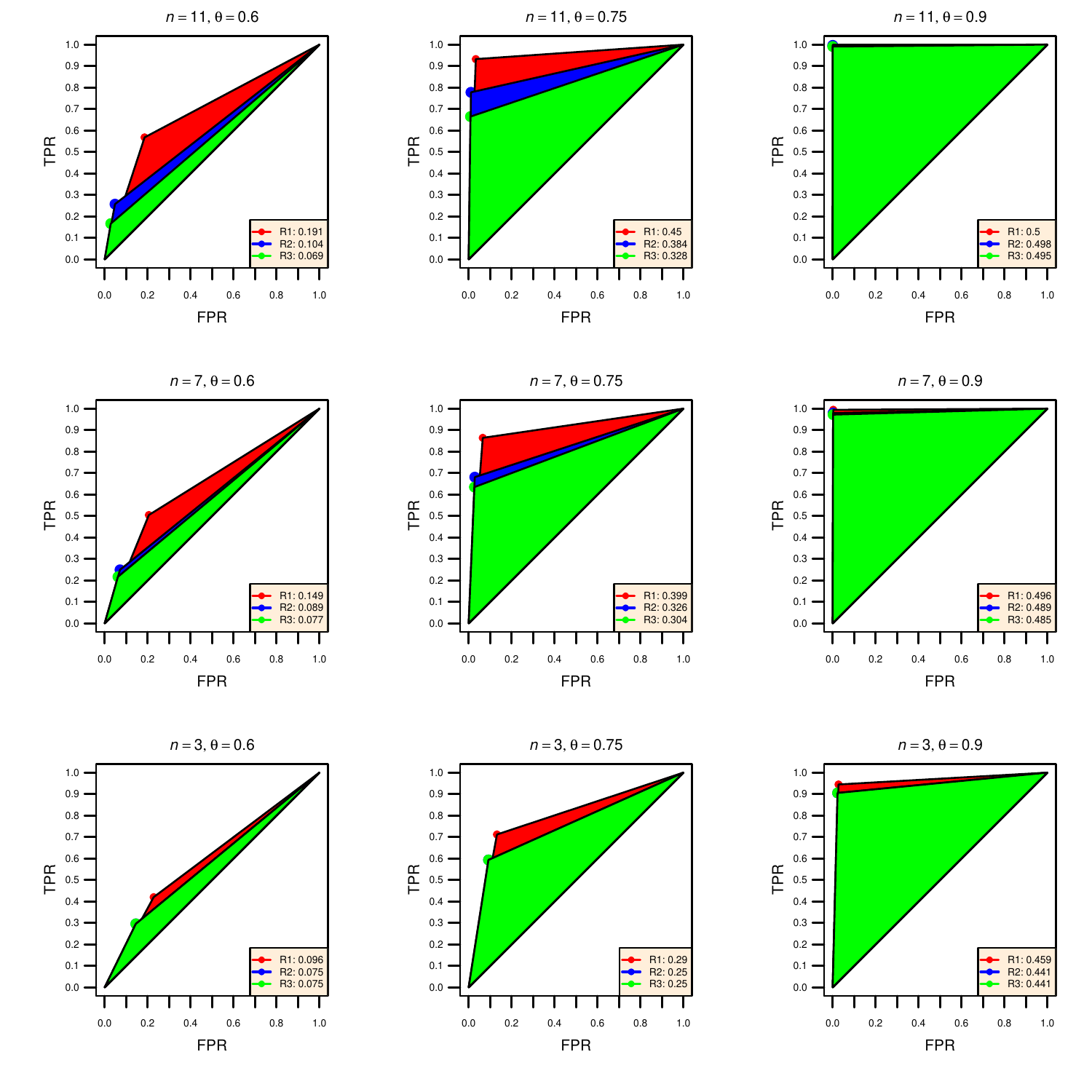} 

\end{knitrout}

\caption{ROC representation of $R_1$, $R_2$ and $R_3$, for several committee sizes $n$ and 
  different competences $\theta$. 
  Values in the legend are the AOT.
  \label{rocs}} 
\end{center}
\end{figure}


The ROC analysis helps in comparing several decision rules, but it is also useful to visualize the performance 
of a given rule depending on the parameters. 
For example,
taking three values for the competence, $\theta=0.60,\, 0.75, \, 0.90$, 
and three values for the committee size 
$n=3,\, 7,\, 11$, 
the ROC representation in Figure \ref{rocn} displays a curve going form near the diagonal  
when $\theta=0.60$ to near the corner (0,1), when $\theta=0.90$. 
From the figure it is apparent that the quality of the voters in terms of the competence 
is definitely much more important than its quantity. (Karotkin and Paroush \cite{Karotkin2003} arrive
to the same conclusion in the single-premise case.)

 {\tiny

}

\begin{figure}[!t]
  \begin{center}
      {\tiny
\begin{knitrout}
\definecolor{shadecolor}{rgb}{0.969, 0.969, 0.969}\color{fgcolor}
\includegraphics[width=\maxwidth]{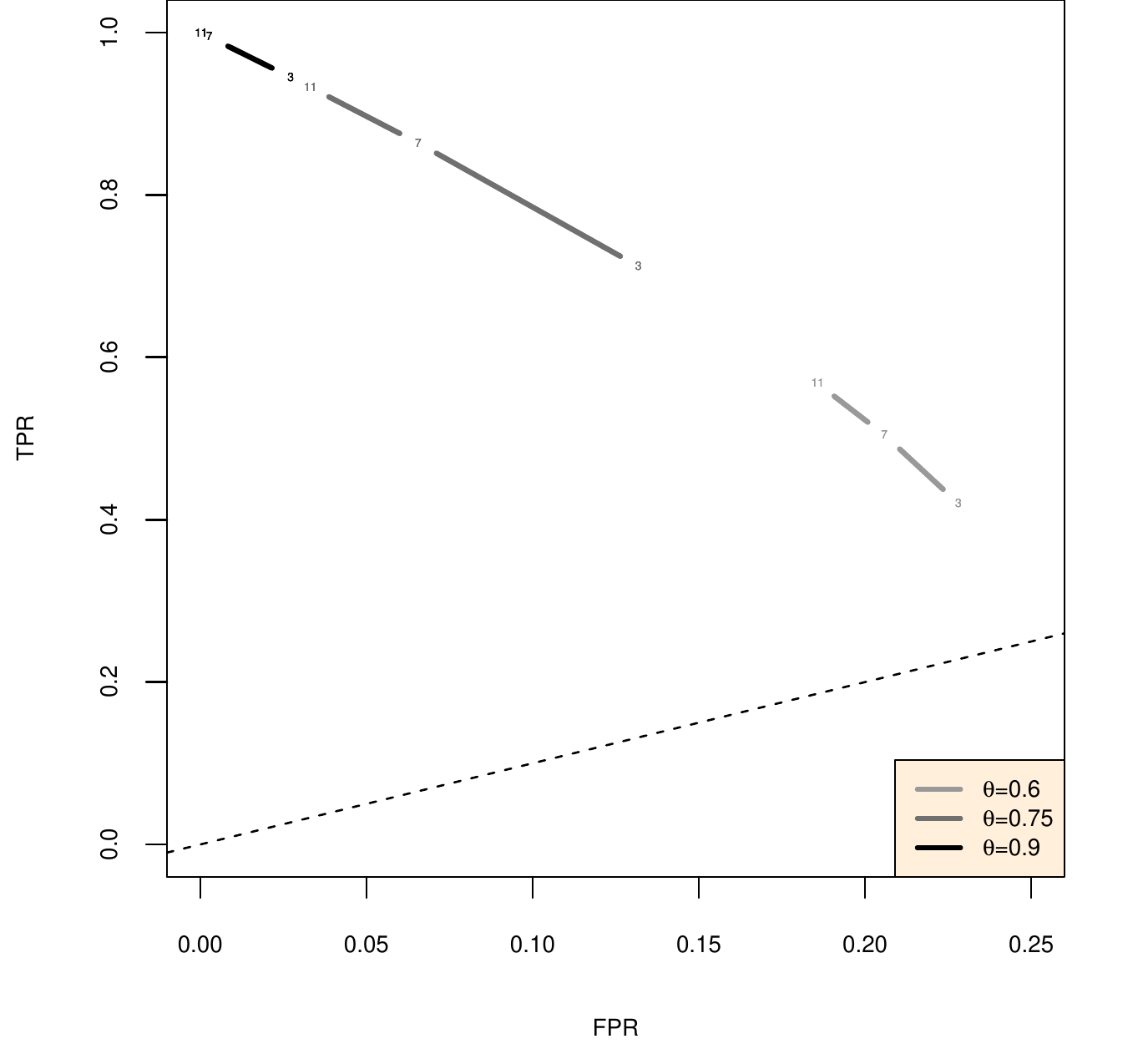} 

\end{knitrout}
    }
    \caption{ROC representation of $R_1$ for different numbers of voters $n=3, 7, 11$ 
      (the number is used as the location in ROC space) and different competences 
      $\theta=0.60, 0.75, 0.90$  
      in different shades of grey. 
      Notice than the range in the horizontal axis has been rescaled.  
     \label{rocn}}
  \end{center} 
\end{figure}

It is also natural to ask, for a given rule and competence value $\theta$, 
what is the minimum number of voters $n$ that ensures that the TPR, TNR and AOT ($w=0.5$) reach a certain 
given threshold $k$? This minimum is easy to find out. For example,
Table \ref{taulan} give the numbers,
for decision rule $R_1$ (IbyI), a range of $\theta$ values 
and quite demanding thresholds.  
Notice that 
to ensure a certain threshold of TPR 
the size of the committee must be greater than
to ensure the same threshold for TNR. 
This is because for a fixed size $n$, rule $R_1$ the probability to produce a true negative is greater than
the probability to produce a true positive

{\tiny

\begin{table}[ht]
\centering
\begingroup\footnotesize
\begin{tabular}{|l|cccccccc|}
  \hline
  & \multicolumn{8}{c|}{$\theta$} \\
   & 0.60 & 0.65 & 0.70 & 0.75 & 0.80 & 0.85& 0.90 & 0.95 \\
 \hline
For TPR, $k$=0.95 & 95 & 41 & 23 & 13 & 9 & 7 & 5 & 3 \\ 
  For TNR, $k$=0.95 & 65 & 29 & 15 & 9 & 7 & 5 & 3 & 3 \\ 
  For AOT, $k$=0.45 & 81 & 35 & 19 & 13 & 7 & 5 & 3 & 3 \\ 
   \hline
\end{tabular}
\endgroup
\caption{For each indicator, TPR, TNR and AOT, 
minimum committee size needed to reach the threshold $k$, for rule $R_1$ 
and different competence values $\theta$.} 
\label{taulan}
\end{table}

}

\section{Discussion}\label{disc}

  In this paper, we have defined
  a theoretical framework, based on a probabilistic model, 
  that describes the behaviour
  of a committee faced to decide the truth or falsity of a compound question.
  The application of the Receiver Operating Characteristics (ROC), a concept originating in 
  signal processing, and adopted in several other fields, seems to be new in 
  judgement aggregation research. It allows an objective assessment of the quality
  of a group decision on a complex issue, based on the (possibly subjective) competence
  of the members of the group. It also allows to compare different decision rules, 
  both for symmetric or asymmetric penalising weights on the false positives and false negatives.     

  The simplicity of the model has allowed us to focus on the  
  methodology of ROC space, but the assumptions can be easily 
  weakened in several ways and the computations can be adapted without much difficulty.
  For example: 
  \newpage 
  \begin{itemize}
    \item 
    \emph{Different voters' competence $\theta$.}
    If competences are different, the law of $(X,Y,Z,T)$ described in Proposition \ref{pr3} is no
    more multinomial. The vote of each committee member $k=1,\dots,n$ is a random
    object $J_k$ equal to one of the four possible vote schemes with 
    certain probabilities $(p_x^k,p_y^k,p_z^k,p_t^k)$, independently. 
    If $\theta_k$ is the competence of member $k$, then, under the true state $P\wedge Q$, 
    this vector of probabilities is 
\begin{equation*}
    (\theta_k^2, \theta_k(1-\theta_k), \theta_k(1-\theta_k), 1-\theta_k^2)\ .
\end{equation*}    
    The law of $(X,Y,Z,T)$ will be the law of $J_1+\cdots+J_n$, which can be computed for
    any given set of parameters $\theta_k$. The same can be done for the states of nature
    $P\wedge\neg Q$, $\neg P\wedge Q$, and $\neg(P \wedge Q)$. 
    
    The true and false positive rates under 
    the vector of competences $\theta=(\theta_1,\dots,\theta_n)$
    and rule $R$ will be given, following Definitions \ref{defi2} and \ref{def:FPR}, by 
    \begin{equation*}
    \text{TPR}(n,\theta,R)
    =
    \sum_{R(i,j,k,\ell)=1} {\mathbb{P}}_{P\wedge Q}\big\{J_1+\cdots+J_n=(i,j,k,\ell) \} \ ,
    \end{equation*}     
    and 
    \begin{equation*}
    \text{FPR}(n,\theta,R)
    =
    \sum_{R(i,j,k,\ell)=1} {\mathbb{P}}_{P\wedge \neg Q}\big\{J_1+\cdots+J_n=(i,j,k,\ell) \} \ .
    \end{equation*} 
    
A study of dichotomous decision making under different individual competences, within a
probabilistic framework, can be found in Sapir \cite{Sapir1998}.

  \item 
  \emph{Non-independence between voters.} 
    If the committee members do not vote independently, 
    then, in order to make exact computations, one must have the joint probability law of 
    the vector $(J_1,\dots,J_n)$, which take values in the 
    $n$-fold Cartesian product of $\{P\wedge Q, P\wedge \neg Q, \neg P\wedge Q, \neg(P \wedge Q)\}$.
    From the joint law, the distribution of the sum $J_1+\cdots+J_n$ can always be made explicit,
    for each state of nature and taking into account the given vector of competences $(\theta_1,\dots,\theta_n)$.
    And from there, the values of FPR, FNR and AOT can be also obtained for any rule. 
  
  Boland \cite{10.2307/2348873} studies this situation of non-independence and diverse competence values
  for the voting of a single question, and assuming the existence of a ``leader'' in the committee. He 
  generalises Condorcet theorem when the correlation coefficient between voters does not exceeds a certain
  threshold. That situation is completely different from ours.
  
  \item 
  \emph{Non-independence between the premises.}
  The premises may depend on each other in the
  sense that believing that $P$ is true or false changes the perception
  on the veracity or falsity of $Q$. An extreme example of dependence is the classical
  $P=$``existence of a contract'' and $Q=$``defendant breached the contract'', where
  voting $\neg P$ forces to vote $\neg Q$. 
  
  In these situations, some more data is needed, namely, the competence of each 
  voter on one of the premises alone, and on the other premise conditioned to 
  have guessed correctly the first, and conditioned to have guessed it incorrectly. To wit,
  suppose that $\theta_P$ is the competence of a voter on premise $P$, that
  $\theta_{Q|P}$ is her competence on $Q$ if she guesses correctly on $P$, and that
  $\theta_{Q|\bar P}$ is her competence on $Q$ assuming she does not guess
  correctly on $P$. Then, the first row in the table of Proposition \ref{pr3} would read
  
{\centering
  \setlength\extrarowheight{2pt}
  \begin{tabular}{r|c|c|c|c}
             & $p_x$ & $p_y$ & $p_z$ & $p_t$ \\
  \hline  
  $P\wedge Q$ & $\theta_P \theta_{Q|P}$  &  $\theta_P (1-\theta_{Q|P})$  &  $(1-\theta_P) \theta_{Q|\bar P}$  &  $(1-\theta_P) (1-\theta_{Q|\bar P})$ \\
  \hline   
\end{tabular}     
\par }  
 and similarly for the other true states of nature.
  \item 
  \emph{Competence depending on the true state.} 
  The probability to guess
  the truth may depend on the truth itself,
  i.e $\theta=(\theta_P, \theta_{\neg P}, \theta_{Q}, \theta_{\neg Q})$ could be four
  different parameters associated to the committee members, giving the probabilities of
  guessing the true state of nature when $P$ is true, $P$ is false, $Q$ is true and $Q$ is 
  false, respectively. And these probabilities can be different for each individual,
  of course.
  
\end{itemize}   

All these extensions can be combined together. The key to compute FPR, FNR,
and consequently the values of AOT and $\text{WAOT}_w$ of any given rule is the ability
to compute the probability of appearance of all possible contingency tables (Table \ref{tab3});
and this is possible for all the extensions of the list above. It is not easy to establish
general theorems of comparison between rules, but a computer software can 
evaluate and compare rules for any specific value of all the parameters involved.
Notice that it is not even necessary to fix $n$ and $\theta$. Two 
differently formed committees, adhering to 
the same or to different decision rules, can be compared applying the same ideas, 
using the symmetric
or the weighted area of the triangle.
In this paper, we have stuck to the simplest of the situations in our exposition, to better highlight
the methodology, and to obtain some specific theoretical results.

In ROC analysis one usually relies in sample training data for the classifier. 
Here we have postulated the existence of an exact competence parameter $\theta$.
This parameter can be assigned on subjective grounds, but of course past data, if available,
can be used to estimate its value. Furthermore, if after a new experiment it is possible
to assess the quality of a voter's decision, this value could be readjusted in a
Bayesian manner. As examples of possible practical relevance, we mention: 
In court cases, the level of a court, and the proportion of cases
that have been successfully appealed to higher instances, can measure the competence 
of individual judges;
in some professional sport competitions, referees are ranked according to their performance in past events,
and it is usually easy to determine a posteriori the proportion of their correct 
decisions in a given event; in simultaneous medical tests, the  ``competence'' or reliability 
of each one is usually known to some extent, and they can be combined to offer the best
diagnostic decision, taking possibly into account the risks of false positives or 
false negatives through the weight parameter $w$.


It may be argued that the explicit appearance of the doctrinal paradox 
is rare in practice, but this depends
on the competence parameter $\theta$, the committee size $n$, and the 
two specific rules that we are comparing.
In fact, the probability that the individual
decisions lead to the paradoxical situation between two given rules can be 
computed explicitly. For instance,   
List \cite[Propostion 2]{List2005},
computes the probability of the occurrence of the paradox from our formula (\ref{mult}),
and the two classical rules issue-by-issue (IbyI, $R_1$) and case-by-case (CbyC, $R_2$).

However, the potential appearance of the paradox
cannot be avoided, 
except for the trivial one-member committee.
The interest must then be focused in the choice of the ``best decision rule'' 
among a catalogue of rules. 
The question then becomes to define a criterion to evaluate decision rules. 
In this work, we have considered a family of criteria, which are completely objective, 
once fixed the subjective weight of the two competing risks,
the false positive and 
the false negative, in deciding if the conjunction $P\wedge Q$ is true.
This is contrast with the classical theory of Hypothesis Testing, but in line with Statistical 
Decision Theory.

The subjective evaluation of the risks is clearly an a priori matter,
to be fixed before the casting of the votes, and therefore any instance of
the situation we are analysing has a definite logical way to arrive to the 
correct decision.   
In this sense, we dare claim to have a complete solution to the doctrinal paradox. 
To be thoroughly honest, though, we must mention that the definition of the false positive
rate (Definition \ref{def:FPR}), 
justified by Proposition \ref{pr6}, has some degree of arbitrariness.


Let us finish by commenting on possible extensions and open questions:   

Instead of using the majority principle, another ``qualified majority'' can be employed, and the analysis
of the modified rules can be done similarly. 
Still another possibility is to use \emph{score versions}
of the rules: Instead of 0/1 outputs, one can use more general mappings from the set of tables into the set
of real numbers. These mappings are called \emph{scores}. 
The natural scores to associate to rules $R_1$, $R_2$, $R_3$ are,
respectively 
\begin{equation*}
S_1= x+y\wedge z-m\ ,\quad S_2=x+\lfloor \frac{y\wedge z}{2}\rfloor-m\ ,\quad S_3=x-m\ ,
\end{equation*}
and clearly $S_3\le S_2\le S_1$. ROC analysis can be done the same with scores, where the goal
is to find the best score to fix as a boundary in deciding between $P\wedge Q$ and $\neg(P\wedge Q)$
from the point of view of the \emph{area under the curve}. Admissible scores can be
defined similarly to admissible rules: they must be non-decreasing functions 
with respect to the partial order on ${\Bbb T}$ defined in Section \ref{3rules}. 
See \cite{Fawcett:2006:IRA:1159473.1159475} for a good short introduction to
score rules and ROC curves.

We have compared three particular rules, two of them 
traditional, and a reasonable third one that lies in between. They are easily expressed in terms
of the entries of the contingency table of votes. But there are much more admissible rules;
a total of 36 in the case $n=3$, and they are not easy to enumerate systematically in general.
Hence, the question of enumerating all admissible rules and choosing the best according to some
criterion is open. 

Finally, notice that
in Theorem \ref{teor12} we have not assumed that the function $\theta\rightarrow D(\theta)$ is increasing. 
We conjecture that it always is, but we do not have yet a formal proof. If this conjecture is true, or in a
particular case in which it is checked to be true, then one could strengthen the theorem by claiming that 
for every weight $\frac{1}{2}<w<1$, there exists a value $C(w)$, smaller than $w$, such that
\begin{align*}
\theta> C(w) &\Rightarrow \text{\rm WAOT}_w(n,\theta,R_1) > \text{\rm WAOT}_w(n,\theta,R_2) \\
\theta< C(w) &\Rightarrow \text{\rm WAOT}_w(n,\theta,R_1) < \text{\rm WAOT}_w(n,\theta,R_2) 
\end{align*}
(see Figure \ref{fig:D(theta)inc}, and compare with Figure \ref{fig:D(theta)}), 
and similarly for $R_2$ in relation with $R_3$.  

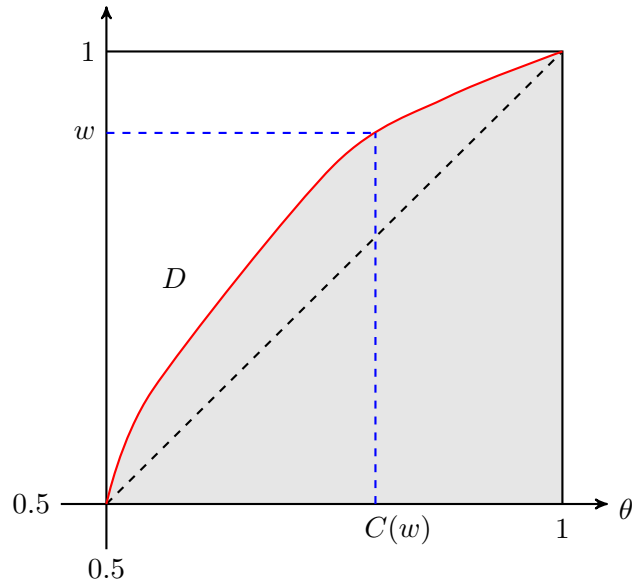
\begin{figure}[h] 
  \centering 
  \begin{tikzpicture}[
  thick,
  >=stealth',
  scale=6,
  ]
  \coordinate (O) at (0,0);
     \fill [gray!20] 
          (0, 0)
          -- plot[smooth] coordinates {(0,0) (0.10,0.25) (0.50,0.75) (0.75,0.90) (1,1)}
          -- (1, 0) 
          -- cycle;
  
  \draw[->] (-0.1,0) -- (1.1,0) ;
  \draw[->] (0,-0.1) -- (0,1.1);
  \draw[-] (1,0) -- (1,1) -- (0,1);
  \node [right] at (1.10,-0.01) {$\theta$};
  \node [left] at (0.00, 0.82) {$w$};
  \node [left] at (0.20,0.50) {$D$};
  \node [below] at (0.64,0) {$C(w)$};
  \node [left] at (-0.1, 0) {$0.5$};
  \node [below] at (0, -0.1) {$0.5$};
  \node [below] at (1.00,-0.01) {$1$};
  \node [left] at (-0.00,1) {$1$};
  \draw [-, dashed] ( 0,0) -- ( 1, 1); 
  \draw[red] plot[smooth] coordinates {(0,0) (0.10,0.25) (0.50,0.75) (0.75,0.90) (1,1)};
  \draw [-, thick, blue, dashed] ( 0,0.82) -- ( 0.59, 0.82);  
  \draw [-, thick, blue, dashed] ( 0.59,0) -- ( 0.59, 0.82);          
  \end{tikzpicture}   
  \caption{
  If $D$ is a non-decreasing function, then $C(w)$ splits the $\theta$-axis in two segments. For $\theta$ on the right, 
  rule $R_1$ is better than $R_2$ in the sense of the weighted area under the triangle with weight $w$; for $\theta$
  on the left, $R_2$ is better than $R_1$.
  Compare with the general $D$ of Figure \ref{fig:D(theta)}.
  \label{fig:D(theta)inc}}
\end{figure}



%

\section*{Acknowledgements}
  This work has been partially supported by grant numbers MTM2014-59179-C2-1-P    
  from the Ministry of Economy and Competitiveness of Spain, and 2017-SGR-1094 
  from the Ministry of Business and Knowledge of Catalonia. 

\FloatBarrier

\section*{Appendix} 
 
\paragraph{Proof of Proposition  \ref{pr1}}

  For rule $R_1$, we must show the equivalence
\begin{equation*}
(x+y>z+t,\ x+z>y+t)\ \Longleftrightarrow\ x>m-y\wedge z\ . 
\end{equation*}   
Indeed,  
\begin{align*}
    x+y>z+t\ &\Longleftrightarrow\  x+y>n-x-y\ 
     \Longleftrightarrow\   2x+2y> n=2m+1 \\
   &\Longleftrightarrow\ x+y> m+\tfrac{1}{2}\  
     \Longleftrightarrow\   x>m-y\ , 
\end{align*}
taking into account that we are dealing with integer numbers.
The equivalence $x+z>y+t \Longleftrightarrow  x>m-z$ is proved in an identical way
and (\ref{r12}) follows.

 For rule  $R_2$, we must prove the equivalence in (\ref{r22}), that is,
\begin{equation*}
(x>z+t,\ x>y+t)\ \Longleftrightarrow\ x>m-\lfloor\tfrac{y\wedge z}{2}\rfloor\ .
\end{equation*} 
 We have
\begin{align*}
x>z+t\ \Longleftrightarrow\ x>n-x-y\  
     \Longleftrightarrow\   2x>  2m+1-y\
     \Longleftrightarrow\ x > m-\tfrac{y-1}{2}\ . 
\end{align*}
If $y$ is odd,  
 say $y=2q+1$, then both $\frac{y-1}{2}$ and $\lfloor\frac{y}{2}\rfloor$ are equal to $q$,
 and we get the equivalence with $x>m-\lfloor\tfrac{y}{2}\rfloor$.
 If $y=2q$,
$$
 x>m-\tfrac{y-1}{2}
 = m-q+\tfrac{1}{2}\ \Longleftrightarrow\  x>m-q=m-\lfloor\frac{y}{2}\rfloor \ .
 $$
Analogously, we obtain 
$x>y+t    \Longleftrightarrow      x>m-\lfloor\frac{z}{2}\rfloor,$
and (\ref{r22}) is proved.

Finally, for rule  $R_3$, we have to prove the equivalence in (\ref{r32}):
\begin{equation*} 
x>y+z+t\  \Longleftrightarrow\ x>m  \ .
\end{equation*}
Indeed, 
\begin{align*}
x>y+z+t\ \Longleftrightarrow\ x>n-x\  
         \Longleftrightarrow\   2x>  2m+1\
         \Longleftrightarrow\ x > m + \tfrac{1}{2}\   
          \Longleftrightarrow\   x>m \ ,
\end{align*}
and we are done. \qed

\paragraph{Proof of Proposition  \ref{pr2}} 
The admissibility of the three rules is completely obvious from the defining 
inequalities (\ref{r1}), (\ref{r2}) and (\ref{r3}). The order 
$R_3\le R_2 \le R_1$  as functions $\mathbb{T}\longrightarrow \{0,1\}$  
is also clear, so we focus on the equalities and strict inequalities. 

First, we will prove that  $R_2 = R_3 $ for $n=3$ (i.e. $m=1$). 
When $n=3$,  $y+z \leq 3$ and this implies $\lfloor\frac{y\wedge z}{2}\rfloor = 0$. 
According to Proposition \ref{pr1}, rules $R_2$ and $R_3$ are the same.

Secondly, for $n=5$ ($m=2$),  we have $y+z \leq 5$ and then $\lfloor\frac{y\wedge z}{2}\rfloor \leq 1$.
On the points $(x,y,z,t)$ for which $\lfloor\frac{y\wedge z}{2}\rfloor = 0$, the  two rules $R_2$ and $R_3$
clearly coincide, by Proposition \ref{pr1}. Assume $\lfloor\frac{y\wedge z}{2}\rfloor = 1$. That means $y\ge 2$, $z\ge 2$, and 
therefore $x\le 1$. For such points, the value of both rules is zero. Hence, they coincide everywhere.

To see that $R_1>R_2$ for all  $n\geq 3$, recall that this inequality means that $R_1\ge R_2$ and 
the two rules are not identical. Consider the point $(x,y,z,t)=(1,m,m,0)$. On this point,
condition (\ref{r1}) is satisfied, but (\ref{r2}) fails. Therefore, $R_1(1,m,m,0)=1$ and
$R_2(1,m,m,0)=0$. 

To show that  $R_2>R_3$ for all  $n\geq 7$ (implying that $m\geq 3$), 
we distinguish two cases, depending on whether $m$ is even or odd. For an even $m\geq 4$, 
take $(x,y,z,t)=(m,\frac{m}{2},\frac{m}{2},1)$; 
 for $m\geq 3$ odd, 
take $(x,y,z,t)=(m,\frac{m+1}{2},\frac{m+1}{2},0)$. In both cases condition (\ref{r2}) is satisfied and (\ref{r3}) fails. \qed
 
\paragraph{Proof of Proposition \ref{pr3}}
For each voter $k=1,\dots,n$, let $X_k$ be the random variable taking values in one of the four possible
final decisions of the voter, and let $p_x, p_y, p_z, p_t$ the probability of each of them.
Since these $n$ variables are mutually independent and identically distributed, the law of the counts $(X,Y,Z,T)$ is 
multinomial with parameters $(n,p_x, p_y, p_z, p_t)$. The independence of the decisions concerning
the two premises, gives immediately the particular parameters of the table, depending on the
true state of nature.

\paragraph{Proof of Proposition \ref{pr5}}
The probability of deciding $P\wedge Q$ for any decision rule can be computed summing up the multinomial 
probability function (\ref{mult}) over the set where the rule takes the value 1. That means,
taking into account Proposition \ref{pr1},
 
\begin{align*}
\mathbb{P}\big\{R_1(X,Y,Z,T)=1 \big\} &= 
\sum_{i=0}^n   \sum_{j=0}^{n-i}  \sum_{k=m- i\wedge j+1}^{n-i-j} 
\frac{n!}{k!j!i!(n-i-j-k)!} p_x^k p_y^j p_z^i p_t^{n-i-j-k} \ . 
\\
 \mathbb{P}\big\{R_2(X,Y,Z,T)=1 \big\} &=
\sum_{i=0}^n  \sum_{j=0}^{n-i}  \sum_{k=m-\lfloor\frac{1}{2}i\wedge j\rfloor +1}^{n-i-j} 
\frac{n!}{k!j!i!(n-i-j-k)!} p_x^k p_y^j p_z^i p_t^{n-i-j-k}\ . 
\\
\mathbb{P}\big\{R_3(X,Y,Z,T)=1 \big\} &=    
\sum_{i=0}^n  \sum_{j=0}^{n-i}  \sum_{k=m  +1}^{n-i-j} 
\frac{n!}{k!j!i!(n-i-j-k)!} p_x^k p_y^j p_z^i p_t^{n-i-j-k} 
\\
&=\sum_{k=m+1}^{n} {\textstyle\binom{n}{k}} p_x^k(1-p_x)^{n-k} \ .
\end{align*}

now using the values of Proposition \ref{pr3} for the case $P\wedge Q$, we obtain 
the stated true positive rates.

\paragraph{Proof of Proposition \ref{pr6}} 

To check the equality in (\ref{eqpr6-a}), first recall that  
every admissible decision rule satisfies $R(x,y,z,t)=R(x,y,z,t)$. Moreover, 
the law of the random vector $(X,Y,Z,T)$ under $P\wedge \neg Q$, which is multinomial 
$ M(n,\theta(1-\theta),\theta^2,(1-\theta)^2,\theta(1-\theta))$, coincides 
with the law of $(X,Z,Y,T)$ under $\neg P\wedge Q$. Therefore, 
\begin{equation*}
\mathbb{P}_{P\wedge \neg  Q}\{R(X,Y,Z,T)=1\} = \mathbb{P}_{\neg P\wedge Q}\{R(X,Z,Y,T)=1\} = \mathbb{P}_{\neg P\wedge Q}\{R(X,Y,Z,T)=1\}\ .
\end{equation*}
  
For the inequality in the statement, we have, from Proposition \ref{pr3},
\begin{align*}
\mathbb{P}_{ P\wedge \neg  Q}\{R_1(X,Y,Z,T)=1\}&= \sum_{i=0}^n  \sum_{j=0}^{n-i} \sum_{k=m- i\wedge j+1}^{n-i-j}
{n\choose i,j,k,\ell} \theta^{n-i+j} (1-\theta)^{n+i-j}\ ,  
\\
\mathbb{P}_{\neg P\wedge \neg  Q}\{R_1(X,Y,Z,T)=1\} &=\sum_{i=0}^n  \sum_{j=0}^{n-i} \sum_{k=m- i\wedge j+1}^{n-i-j}
{n\choose i,j,k,\ell} \theta^{2n-i-j-2k} (1-\theta)^{i+j+2k} \ ,
\end{align*}
where $\ell=n-i-j-k$. 
Each term in the first line is greater than the corresponding one in the second line: 
$\theta^{n-i+j}(1-\theta)^{n+i-j}>\theta^{2n-i-j-2k}(1-\theta)^{i+j+2k}$ is equivalent to $\theta^{2k+2j-n} > (1-\theta)^{2k+2j-n}$, and  
this inequality, since $\theta>\frac{1}{2}$, is true provided that $2k+2j-n>0$. Indeed, 
\begin{align*}
k\geq m-i\wedge j +1\geq m-j+1 \ \Rightarrow\  
2k + 2j \geq 2m+2\ \Rightarrow\ 2k+2j-n\geq 1 \ .
\end{align*} 
The same argument is valid for $R_2$ and $R_3$, since $m-i\wedge j+1 \le m-\lfloor \frac{i\wedge j}{2}\rfloor +1 \le m+1$. This completes the proof.

\emph{Remark:} In fact, one can prove, using a probabilistic coupling argument, that the inequality 
$\mathbb{P}_{P\wedge \neg Q}\{R=1\} \ge \mathbb{P}_{\neg P\wedge \neg Q}\{R=1\}$ holds true not only for 
$R_1, R_2, R_3$, but for
any admissible rule as defined in \ref{def:AdmisRule}. For simplicity, we have restricted ourselves 
here to prove the statement as it is. \qed

\paragraph*{R session info.}\label{rinfo}
%
Computations have been done 
in R with the following setup: \\ 
\raggedright
  R version 3.4.1 (2017-06-30), \verb|x86_64-w64-mingw32|. \\
  Base packages: base, datasets, graphics, grDevices, stats, utils. \\
  Other packages: knitr~1.17, xtable~1.8-2, compiler~3.4.1,
  digest~0.6.12, evaluate~0.10.1, magrittr~1.5, methods~3.4.1,
  stringi~1.1.6, stringr~1.3.0, tools~3.4.1.

\bibliographystyle{plain}

\begin{thebibliography}{10}

\bibitem{10.2307/2348873}
Philip~J. Boland.
\newblock Majority systems and the condorcet jury theorem.
\newblock {\em Journal of the Royal Statistical Society. Series D (The
  Statistician)}, 38(3):181--189, 1989.

\bibitem{Bonnefon2010}
Jean-Fran{\c{c}}ois Bonnefon.
\newblock Behavioral evidence for framing effects in the resolution of the
  doctrinal paradox.
\newblock {\em Social Choice and Welfare}, 34(4):631--641, Apr 2010.

\bibitem{Bovens2004}
Luc Bovens and Wlodek Rabinowicz.
\newblock Voting procedures for complex collective decisions: an epistemic
  perspective.
\newblock {\em Ratio Juris}, 17(2):241--258, 2004.

\bibitem{Bovens2006}
Luc Bovens and Wlodek Rabinowicz.
\newblock Democratic answers to complex questions -- an epistemic perspective.
\newblock {\em Synthese}, 150(1):131--153, May 2006.

\bibitem{Camps2012}
Rosa Camps, Xavier Mora, and Laia Saumell.
\newblock A general method for deciding about logically constrained issues.
\newblock {\em Annals of Mathematics and Artificial Intelligence},
  64(1):39--72, 2012.

\bibitem{nla.cat-vn2403603}
Jean-Antoine-Nicolas de~Caritat Condorcet.
\newblock {\em Essai sur l'application de l'analyse a la probabilite des
  decisions rendues a la pluralite des voix}.
\newblock Imprimerie royale Paris, 1785.

\bibitem{deClippel201534}
Geoffroy de~Clippel and Kfir Eliaz.
\newblock Premise-based versus outcome-based information aggregation.
\newblock {\em Games and Economic Behavior}, 89:34 -- 42, 2015.

\bibitem{Dietrich2007}
Franz Dietrich and Christian List.
\newblock Arrow's theorem in judgment aggregation.
\newblock {\em Social Choice and Welfare}, 29(1):19--33, Jul 2007.

\bibitem{dietrich_list_2007}
Franz Dietrich and Christian List.
\newblock Strategy-proof judgment aggregation.
\newblock {\em Economics and Philosophy}, 23(3):269–300, 2007.

\bibitem{Egan1975}
James~P. Egan.
\newblock {\em Signal detection theory and ROC-analysis}.
\newblock Academic Press, 1975.

\bibitem{Fawcett:2006:IRA:1159473.1159475}
Tom Fawcett.
\newblock An introduction to {ROC} analysis.
\newblock {\em Pattern Recognition Letters}, 27(8):861--874, June 2006.

\bibitem{HandTill2001}
David~J. Hand and Robert~J. Till.
\newblock A simple generalisation of the area under the {ROC} curve for
  multiple class classification problems.
\newblock {\em Machine Learning}, 45(2):171--186, 2001.

\bibitem{Karotkin2003}
Drora Karotkin and Jacob Paroush.
\newblock Optimum committee size: Quality-versus-quantity dilemma.
\newblock {\em Social Choice and Welfare}, 20(3):429--441, Jun 2003.

\bibitem{Kornhauser1992169}
Lewis~A. Kornhauser.
\newblock Modeling collegial courts {I}: Path-dependence.
\newblock {\em International Review of Law and Economics}, 12(2):169 -- 185,
  1992.

\bibitem{Kornhauser1993}
Lewis~A. Kornhauser and Lawrence~G. Sager.
\newblock The one and the many: Adjudication in collegial courts.
\newblock {\em California Law Review}, 81(1), 1993.

\bibitem{List2005}
Christian List.
\newblock The probability of inconsistencies in complex collective decisions.
\newblock {\em Social Choice and Welfare}, 24(1):3--32, 2005.

\bibitem{List2012}
Christian List.
\newblock The theory of judgment aggregation: an introductory review.
\newblock {\em Synthese}, 187(1):179--207, Jul 2012.

\bibitem{list_pettit_2002}
Christian List and Philip Pettit.
\newblock Aggregating sets of judgments: An impossibility result.
\newblock {\em Economics and Philosophy}, 18(1):89–110, 2002.

\bibitem{List2009-LISJAA}
Christian List and Clemens Puppe.
\newblock Judgment aggregation: A survey.
\newblock In Christian List and Clemens Puppe, editors, {\em Handbook of
  Rational and Social Choice}. Oxford University Press, 2009.

\bibitem{Sapir1998}
Luba Sapir.
\newblock The optimality of the expert and majority rules under exponentially
  distributed competence.
\newblock {\em Theory and Decision}, 45(1):19--36, Aug 1998.

\end{thebibliography}

\end{document}